\documentclass[12pt]{article}
\usepackage{hyperref} 
\hypersetup{colorlinks=true, allcolors=black} 
\usepackage{amsfonts,amsmath,amsxtra,amsthm}
\usepackage{amssymb}
\usepackage{bm}
\usepackage[utf8]{inputenc}
\usepackage{lscape}
\usepackage[normalem]{ulem}
\usepackage[russian,english]{babel}
\usepackage{xcolor}
\usepackage{mathrsfs}

\usepackage{thmtools,thm-restate}

\renewcommand{\epsilon}{\varepsilon}

\def\a{\alpha}

\def\C{\mathbb{C}}
\def\b{\beta}

\def\beq{\begin{equation}}
\def\eeq{\end{equation}}
\def\beqq{\begin{equation*}}
\def\eeqq{\end{equation*}}

\def\bs{\begin{split}}
	\def\es{\end{split}}

\def\bl{{\boldsymbol{\lambda}}}

\def\bo{\boldsymbol{\omega}}

\def\bx{\boldsymbol{x}}
\def\bt{\boldsymbol{t}}
\def\by{\boldsymbol{y}}
\def\bz{\boldsymbol{z}}
\def\bk{\boldsymbol{k}}

\def\bbx{\underline{\boldsymbol{x}}}
\def\bby{\underline{\boldsymbol{y}}}
\def\bbt{\underline{\boldsymbol{t}}}
\def\bbz{\underline{\boldsymbol{z}}}

\def\const{{2\pi\imath}}

\def\cbk{\color{black}}

\def\Im{\operatorname{Im}}
\def\g{\gamma}

\def\K{{K}}
\def\KK{\hat{\K}}
\def\l{\lambda}

\def\o{\omega}

\def\R{\mathbb{R}}
\def\Re{\mathrm{Re}\,}
\def\Res{\operatorname{Res}}

\def\ve{\varepsilon}
\def\vf{\varphi}

\def\Z{\mathbb{Z}}

\newtheorem{lemma}{Lemma}
\newtheorem*{lemma*}{Lemma}

\newtheorem*{proposition*}{Proposition}

\newtheorem*{theorem*}{Theorem}

\newtheorem*{remark*}{Remark}
\newcommand{\rf}[1]{(\ref{#1})}

\newcommand{\LL}[1]{\Lambda_{#1}}

\begin{document}
\-\vspace{-2.5cm}

\begin{center}
{\bf \large Baxter operators in Ruijsenaars hyperbolic system IV. \\[4pt] Coupling constant reflection symmetry}
\bigskip

{\bf  N. Belousov$^{\dagger\times}$, S. Derkachov$^{\dagger\times}$, S. Kharchev$^{\bullet\ast}$, S. Khoroshkin$^{\circ\ast}$
}\medskip\\
$^\dagger${\it Steklov Mathematical Institute, Fontanka 27, St. Petersburg, 191023, Russia;}\smallskip\\
$^\times${\it National Research University Higher School of Economics, Soyuza Pechatnikov 16, \\St. Petersburg, 190121, Russia;}\smallskip\\
$^\bullet${\it National Research Center ``Kurchatov Institute'', 123182, Moscow, Russia;}\smallskip\\
$^\circ${\it National Research University Higher School of Economics, Myasnitskaya 20, \\Moscow, 101000, Russia;}\smallskip\\
$^\ast${\it Institute for Information Transmission Problems RAS (Kharkevich Institute), \\Bolshoy Karetny per. 19, Moscow, 127994, Russia}
\end{center}

\abstract{
 We introduce and study a new family of commuting Baxter operators in the Ruijsenaars hyperbolic system, different from that considered by us earlier. Using a degeneration of Rains integral identity we verify the commutativity between the two families of Baxter operators and explore this fact for the proof of the coupling constant symmetry of the wave function. We also establish a connection between new Baxter operators and Noumi-Sano difference operators. 
}
\tableofcontents
\section{Introduction}
\subsection{Ruijsenaars system and $Q$-operator}
Denote by $T^{a}_{x_i}$ the shift operator
\begin{equation}
T^{a}_{x_i}:=e^{a\partial_{x_i}}, \qquad \left( T^{a}_{x_i} \, f \right)(x_1,\ldots,x_i,\ldots,x_n) = f(x_1,\ldots,x_i+a,\ldots,x_n)
\end{equation}
and define its products for any subset $I\subset[n] = \{1, \dots, n\}$
\begin{equation}
T^{a}_{I,x}=\prod_{i \in I} T_{x_i}^a.
\end{equation}
The Ruijsenaars hyperbolic system \cite{R1} is governed by commuting symmetric difference operators
\begin{equation}
\label{I2}
H_r(\bx_n;g|\bo) = \sum_{\substack{I\subset[n] \\ |I|=r}}
\prod_{\substack{i\in I \\ j\notin I}}
\frac{\sh^{\frac{1}{2}}\frac{\pi}{\o_2}\left(x_i-x_j-\imath g\right)}
{\sh^{\frac{1}{2}}\frac{\pi}{\o_2}\left(x_i-x_j\right)}
\cdot T^{-\imath\o_1}_{I,x}\cdot \prod_{\substack{i\in I \\ j\notin I}}
\frac{\sh^{\frac{1}{2}}\frac{\pi}{\o_2}\left(x_i-x_j+\imath g\right)}
{\sh^{\frac{1}{2}}\frac{\pi}{\o_2}\left(x_i-x_j\right)}
\end{equation}
where $r = 1, \dots, n$.
Here and in what follows we denote tuples of $n$ variables as
 \begin{equation}
	\bm{x}_n = (x_1, \dots,x_n).
\end{equation}
We also consider gauge equivalent Macdonald operators
\begin{equation}
	\label{I2a}
	 M_r:=	M_r(\bx_n;g|\bo) = \sum_{\substack{I\subset[n] \\ |I|=r}}
	\prod_{\substack{i\in I \\ j\notin I}}
	\frac{\sh\frac{\pi}{\o_2}\left(x_i-x_j-\imath g\right)}
	{\sh\frac{\pi}{\o_2}\left(x_i-x_j\right)}
	\cdot T^{-\imath\o_1}_{I,x}.
\end{equation}
Both families of operators are parametrized by three constants: periods $\bm{\omega}=(\omega_1, \omega_2)$ and a coupling constant $g$,
 which originally are supposed to be real positive. The equivalence is established by means of the measure function
\begin{equation}\label{I5}
\mu(\bx_n)=\prod_{\substack{i,j=1 \\ i\not=j}}^n\mu(x_i-x_j),\qquad
\mu(x):=\mu_g(x|\bo)=S_2(\imath x|\bo) S_2^{-1}(\imath x+g|\bo). \end{equation}
Here $S_2(z|\bo)$ is the double sine function, see Appendix \ref{AppA}.
Namely,
\begin{equation}\label{I4}
\sqrt{\mu (\bx_n)} \,
M_r(\bx_n;g|\bo) \, \frac{1}{\sqrt{\mu  (\bx_n)}}=
H_r(\bx_n,g|\bo).
\end{equation}

In this paper, as well as in \cite{BDKK, BDKK2, BDKK3} and unlike the original Ruijsenaars setting, we consider periods $\bo$ and coupling constant $g$ to be complex valued,  assuming that
\begin{equation}\label{I0a} \Re \o_1 > 0, \qquad \Re \o_2 > 0,\qquad 0< \Re g<\Re \o_1+\Re \o_2  \end{equation}
and
\begin{equation}\label{I0b} \nu_g=\Re\frac{ g}{\o_1\o_2}>0.\end{equation}
Denote the dual coupling constant by reflection
\begin{equation}\label{I3b} g \rightarrow g^\ast=\o_1+\o_2-g\end{equation}
and introduce the function $\K(x)$ 
\begin{equation}\label{I6} \K(x):=\K_g(x|\bo)=S_2^{-1}\Bigl(\imath x +\frac{g^\ast}{2}\Big|\bo\Bigr)S_2^{-1}\Bigl(-\imath x+\frac{g^\ast}{2}\Big|\bo\Bigr).\end{equation}
We also frequently use the products of this function
\begin{equation}\label{I6a}
\K(\bx_n,\by_m)=\prod_{i=1}^n \prod_{j = 1}^m \K(x_i-y_j).
\end{equation}
Note that in the notation \rf{I3b} the measure function \rf{I5} can be rewritten as
\begin{equation}\label{I5a} \mu(\bx_n)=\prod_{\substack{i,j=1 \\ i\not=j}}^n S_2(\imath x_i-\imath x_j|\bo) \, S_2(\imath x_i-\imath x_j+g^\ast|\bo)\end{equation}
due to the reflection formula \eqref{A3}. In \cite{BDKK, BDKK2} we studied a family of operators $Q_{n}(\lambda)$ parameterized by $\lambda \in \mathbb{C}$ and called Baxter $Q$-operators. These are integral operators
\begin{equation}\label{I14}
\left( Q_{n}(\lambda) f\right) (\bm{x}_n) =  d_n(g|\bo) \, \int_{\mathbb{R}^n} d\bm{y}_n \, Q(\bm{x}_n, \bm{y}_n; \lambda) f(\bm{y}_n)
\end{equation}
with the kernel
\begin{equation}\label{Qker} Q(\bx_n,\by_{n};\l)= e^{\const \l(\bbx_n-\bby_n)}
\K(\bx_n,\by_{n})\mu  (\by_{n})
\end{equation}
and normalizing constant
\begin{equation}\label{dconst}
d_n(g|\bo) = \frac{1}{n!} \left[ \sqrt{\omega_1 \omega_2} S_2(g|\bo) \right]^{-n}.
\end{equation}
Here and below for a tuple $\bx_n=(x_1,\ldots,x_n)$ we use the notation $\bbx_n$ for the sum of components
\begin{equation}\label{I10a} \bbx_n=x_1+\ldots+x_n.
\end{equation}
In \cite{BDKK} we established the commutativity of Baxter operators with Macdonald operators and of Baxter operators themselves
\begin{align}\label{QMcomm}
Q_{n}(\lambda) \, M_r&= M_r \, Q_{n}(\lambda), \\[6pt]\label{QQcomm}
Q_{n}(\lambda) \, Q_{n}(\rho) &= Q_{n}(\rho) \, Q_{n}(\lambda),
\end{align}
where $r = 1, \dots, n$. The kernels of the operators in both sides of \rf{QQcomm} are analytic functions of $\l,\rho$ in the strip
\begin{equation} \label{QQcond}  |\Im(\lambda - \rho)| < \nu_g. \end{equation}
The commutativity \eqref{QQcomm} holds under assumptions \eqref{I0a}, \eqref{I0b}. For \eqref{QMcomm} we assume in addition
\begin{equation}\label{MQassump} \Re g < \Re \o_2. \end{equation} 

As it was shown for other integrable systems, kernels of Baxter operators in the classical limit give generating functions of Backlund transformations \cite{KS, Skl}. Classical transformation corresponding to the present $Q$-operators was considered in the paper~\cite{HR0}. We also note that for the trigonometric Ruijsenaars and Calogero models Baxter operators of similar type were studied in the works \cite{KMS, GLO}. 

\subsection{Reflection $g \rightarrow g^*$ and $Q^*$-operator}
Commuting the shift operators to the right in Ruijsenaars operators \eqref{I2} we rewrite them in the form 
\begin{equation} \label{HH}
H_r(\bx_n;g) = \sum_{\substack{I\subset[n] \\ |I|=r}}
\prod_{\substack{i\in I \\ j\notin I}}
\frac{\sh^{\frac{1}{2}}\frac{\pi}{\o_2}\left(x_i-x_j-\imath g\right)}
{\sh^{\frac{1}{2}}\frac{\pi}{\o_2}\left(x_i-x_j\right)}
\frac{\sh^{\frac{1}{2}}\frac{\pi}{\o_2}\left(x_i-x_j-\imath g^*\right)}
{\sh^{\frac{1}{2}}\frac{\pi}{\o_2}\left(x_i-x_j-\imath \o_1 - \imath \o_2\right)} \cdot T^{-\imath\o_1}_{I,x}.
\end{equation} 
Here we omit the dependence on periods $\bo$. From \rf{HH} it is clear that these operators are invariant under reflection $g \rightarrow g^* = \o_1 + \o_2 - g$
\begin{equation} H_r(\bx_n;g) = H_r(\bx_n;g^*). \end{equation} 
Since the measure function \eqref{I5} depends explicitly on $g$, the Macdonald operators \eqref{I2a} lack this symmetry. Instead, due to the connection formula \eqref{I4}, they satisfy
\begin{equation}\label{Msymm} M_r(\bx_n; g) = \eta^{-1}(\bx_n) \, M_r(\bx_n; g^*) \, \eta(\bx_n)  \end{equation}
with
\begin{equation}\label{eta} \eta(\bx_n) = \sqrt{ \prod_{\substack{i,j=1 \\ i\not=j}}^n \frac{\mu_g(x_i-x_j)}{\mu_{g^*}(x_i - x_j) }}= \prod_{\substack{i,j=1 \\ i\not=j}}^n S^{-1}_2(\imath x_i - \imath x_j + g), \end{equation}
where we used reflection formula \eqref{A3} for the double sine function. The function $\eta(\bx_n)$ represents $g$-dependent part of the measure function $\mu(\bx_n)$ \eqref{I5}
\begin{equation}\label{mu-eta}
\mu(\bx_n) = \eta(\bx_n) \Delta(\bx_n)
\end{equation}
with the rest part $\Delta(\bx_n)$ being hyperbolic Vandermonde determinant
\begin{equation} \label{hvand} \Delta(\bx_n) = \prod_{\substack{i,j=1 \\ i\not=j}}^n S_2(\imath x_i - \imath x_j) = \prod_{\substack{i,j=1 \\ i<j}}^n 4\sh \frac{\pi(x_i - x_j)}{\o_1} \sh \frac{\pi(x_i - x_j)}{\o_2}, \end{equation}
see the reflection formula \eqref{trig4}.

As well as Macdonald operators, Baxter operators \eqref{I14} depend explicitly on~$g$. For a moment we emphasize it in notation writing $Q_n(\l; g)$. The symmetry \eqref{Msymm} suggests to look at another family of integral operators. Namely, introduce a family of operators $Q^*_n(\l)$ parameterized by $\l \in \mathbb{C}$
\begin{equation}\label{Q*} Q^*_n(\l) = \eta^{-1}(\bx_n) \, Q_n(\l; g^*) \, \eta(\bx_n).\end{equation}
It is given by integral operator
\begin{equation} \bigl( Q^*_n(\l)  f\bigr) (\bx_n) = d_n(g^*) \int_{\R^n} d\by_n \, Q^*(\bx_n, \by_n; \l) f(\by_n) \end{equation}
with the kernel obtained from \eqref{Qker} using \eqref{mu-eta}
\begin{equation} \label{Q*ker} Q^*(\bx_n, \by_n; \l) = \eta^{-1}(\bx_n)\,  e^{2\pi \imath \l ( \bbx_n - \bby_n) } K^*(\bx_n, \by_n) \Delta(\by_n). \end{equation}
Here we denoted 
\begin{equation} K^*(\bx_n, \by_n) = \prod_{i,j = 1}^n K^*(x_i-y_j) \end{equation} 
with $K^*(x)$ being the counterpart of $K(x)$ \eqref{I6} with respect to the reflection $g \rightarrow g^*$
\begin{equation} K^*(x) := K_{g^*}(x|\bo) = S_2^{-1}\Bigl(\imath x + \frac{g}{2} \,\Big|\, \bo \Bigr) S_2^{-1}\Bigl(-\imath x + \frac{g}{2} \,\Big|\, \bo \Bigr).\end{equation}

In the light of the formulas \eqref{Msymm} and \rf{Q*} the commutativity relations for the first $Q$-operator and Macdonald operators \eqref{QMcomm}, \eqref{QQcomm} imply the same relations between Macdonald operators and the second $Q$-operator, as well as for the second $Q$-operators themselves 
\begin{align}\label{Q*Mcomm}
Q^*_{n}(\lambda) \, M_r(\bx_n;g)&= M_r(\bx_n;g) \, Q^*_{n}(\lambda), \\[6pt] \label{Q*Q*comm}
Q^*_{n}(\lambda) \, Q^*_{n}(\rho) &= Q^*_{n}(\rho) \, Q^*_{n}(\lambda).
\end{align}
For the second relation we need assumption analogous to \eqref{I0b}, that is we assume
\begin{equation}\label{nu*} \nu_{g^*} = \Re \frac{g^*}{\o_1\o_2} > 0. \end{equation}
For the first commutativity we in addition need the condition analogous to \eqref{MQassump}
\begin{equation} \Re g^* <  \Re \o_2.\end{equation}
Note that in the case of real constants $\bo,g$ both assumptions \eqref{I0b}, \eqref{nu*} follow from~\eqref{I0a}. Looking at the mentioned above commutativity relations it is natural to expect that the first and the second $Q$-operators also commute with each other.

\begin{restatable}{thm}{QQthm} \label{theorem1}
Under conditions \eqref{I0a}, \eqref{I0b}, \eqref{nu*} the two families of Baxter $Q$-operators commute
\begin{equation} \label{QQ*comm}
Q_n^*(\l) \, Q_n(\rho) = Q_n(\rho) \, Q_n^*(\l).
\end{equation}
The kernels of the operators in both sides are analytic functions of $\l, \rho$ in the
strip
\begin{equation}\label{lrho}
| \Im(\l - \rho) | < \frac{\nu_g + \nu_{g^*}}{2} = \Re \frac{\o_1 + \o_2}{\o_1 \o_2}.
\end{equation}
\end{restatable}

The proof of this theorem is given in Section \ref{sectionQQ*}. The main ingredient of the proof is the  degeneration of the remarkable elliptic integral identity proved by E.~Rains \cite[Theorem 4.1]{Rains1}. A derivation of this degeneration is given in Appendix~\ref{AppB}.

\subsection{Wave function and local relations}
Commutativity relations between both $Q$-operators and Macdonald operators \eqref{Q*Mcomm}, \eqref{Q*Q*comm}, \eqref{QQ*comm} suggest that they should have common eigenfunctions. Eigenfunctions of the Macdonald operators were constructed by M. Hallnäs  and S. Ruijsenaars \cite{HR1}. In \cite{BDKK2} we proved that these eigenfunctions diagonalize the operator $Q_n(\l)$. In the present work we show that they also diagonalize the operator $Q_n^*(\l)$. It is done using certain local relations between Baxter operators and their degenerations.

Denote by $\LL{n}(\l)$ the integral operator
\begin{equation}\label{I11}\begin{split} \left(\LL{n}(\l)f\right)(\bx_n)=d_{n - 1}(g) \int_{\R^{n-1}}d\by_{n-1} \, \Lambda(\bx_n,\by_{n-1};\l) f(\by_{n-1})
\end{split}\end{equation}
with the kernel
\begin{equation}\label{Lker} \Lambda(\bx_n,\by_{n-1};\l)= e^{\const \l(\bbx_n-\bby_{n-1})}
\K(\bx_n,\by_{n-1}) \mu  (\by_{n-1})
\end{equation}
and normalizing constant $d_{n - 1}(g)$ given by the formula \eqref{dconst}. We call it raising operator. The wave function is given by the multiple integral
\begin{equation} \label{I12}\Psi_{\bl_n}(\bx_n):=\Psi_{\bl_n}(\bx_n; g|\bo)=\LL{n}(\l_n) \, \LL{n-1}(\l_{n-1})\cdots \LL{2}(\l_2) \, e^{\const \l_1x_1}
\end{equation}
or recursively
\begin{equation}\label{Psi-rec}
	\Psi_{\bm{\lambda}_n}(\bm{x}_n) = \Lambda_{n}(\lambda_n) \, \Psi_{\bm{\lambda}_{n - 1}}(\bm{x}_{n - 1}), \qquad \Psi_{\lambda_1}(x_1) = e^{{2\pi \imath}{} \lambda_1 x_1}.
\end{equation}
M.~Hallnäs  and S. Ruijsenaars proved \cite{HR1} for real periods $\bo$ and complex valued constant $g$ that the function \rf{I12}  
is a joint eigenfunction of Macdonald operators. In \cite{BDKK2} by similar arguments we extended this result to the case of complex $\bo$ and proved that it is also an eigenfunction of the first $Q$-operator \rf{I14}. Let us describe the corresponding eigenvalues. \cbk For any $a\in \mathbb{C}$ denote
\begin{equation} \hat{a} = \frac{a}{\o_1 \o_2}, \end{equation}
so that 
\begin{equation} \hat{\bo} = (\o_2^{-1}, \o_1^{-1}), \qquad \hat{g} = \frac{g}{\o_1\o_2}, \end{equation}
and analogously to $g^* = \o_1 + \o_2 - g$ we define
\begin{equation} \hat{g}^* = \hat{\o}_1 + \hat{\o}_2 - \hat{g} = \frac{g^*}{\o_1 \o_2}. \end{equation}
Also introduce the function
\begin{equation} \label{Khat} \hat{K}(\l) := K_{\hat{g}^*}(\l | \hat{\bo}) = S_2^{-1} \Bigl( \imath \l + \frac{\hat{g}}{2} \,\Big|\, \hat{\bo} \Bigr) \, S_2^{-1} \Bigl( -\imath \l + \frac{\hat{g}}{2}  \,\Big|\, \hat{\bo} \Bigr).\end{equation}
Then the spectral description of the first $Q$-operator has the following form
\begin{equation} Q_n(\l) \, \Psi_{\bl_n}(\bx_n) = \prod_{j = 1}^n \hat{K}(\l - \l_j) \, \Psi_{\bl_n}(\bx_n), \end{equation}
see \cite[Theorem 4]{BDKK2}. For the proof we use the iterative structure of the wave function \eqref{I12} and the local exchange relation between $Q$-operator and raising operator \cite[Theorem 3]{BDKK2}
\begin{equation} \label{QLcomm} Q_n(\l) \, \Lambda_n(\rho) = \hat{K}(\l - \rho) \, \Lambda_n(\rho) \, Q_{n - 1}(\l). \end{equation}
In the course of proof one also needs to justify convergence of various multiple integrals appearing on the way.

The $\Lambda$-operator \rf{I11} can be obtained in the certain limit from the first Baxter $Q$-operator, due to the similarity of their kernels \eqref{Qker}, \eqref{Lker} and asymptotic behavior of the function~$K(x)$. Consequently, the exchange relation \eqref{QLcomm} can be obtained in the limit of the $Q$-commutativity relation \eqref{QQcomm}, see \cite[Section~2]{BDKK2}.

In  present work we investigate properties of the second $Q$-operator \eqref{Q*} precisely in the same way. First, taking certain limit of the commutativity relation \eqref{QQ*comm} we obtain an analogous to \eqref{QLcomm} exchange relation, but for the second $Q$-operator. Denote
\begin{equation} \hat{K}^*(\l) := K_{\hat{g}}(\l | \hat{\bo}) = S_2^{-1} \Bigl(\imath \l + \frac{\hat{g}^*}{2} \,\Big|\, \hat{\bo} \Bigr) \, S_2^{-1} \Bigl( -\imath \l + \frac{\hat{g}^*}{2} \,\Big|\, \hat{\bo} \Bigr). \end{equation}
It is a counterpart of the function $\hat{K}(\l)$ with respect to the reflection $\hat{g} \rightarrow \hat{g}^*$. In what follows we always assume conditions \eqref{I0a}, \eqref{I0b} and \eqref{nu*}.

\begin{restatable}{thm}{QLthm} \label{QLthm}
The operator identity
\begin{equation} \label{Q*Lcomm}
Q^*_{n}(\lambda) \, \Lambda_{n}(\rho) = \hat{K}^*(\lambda - \rho) \; \Lambda_{n}(\rho) \, Q^*_{n-1}(\lambda)
\end{equation}
holds true for $\lambda, \rho \in \mathbb{C}$ such that 
\begin{equation} | \Im(\lambda - \rho) | < \frac{\nu_{g^*}}{2}. \end{equation}
\end{restatable}
The  proof is given in Section \ref{sec:exch}.

A further reduction of the formula \eqref{Q*Lcomm} gives one more local relation. Define  the raising operator $\Lambda^*_n(\l)$ analogously to the operator $Q_n^*(\l)$ \eqref{Q*}
\begin{equation}\label{L*def}
	\Lambda^*_n(\l) = \eta^{-1}(\bx_n) \, \Lambda_n(\l; g^*) \, \eta(\bx_{n - 1}),
\end{equation}
where from the right we emphasized the dependence on the coupling constant in the notation of the $\Lambda$-operator. Explicitly, we introduced the integral operator
\begin{equation}
	\bigl( \Lambda^*_n(\l) f \bigr) (\bx_n) = d_{n - 1}(g^*) \int_{\R^{n - 1}} d\by_{n - 1} \, \Lambda^*(\bx_n, \by_{n - 1}; \l) f(\by_{n - 1})
\end{equation}
with the kernel
\begin{equation}\label{L*ker}
	\Lambda^*(\bx_n, \by_{n - 1}; \l) = \eta^{-1}(\bx_n) \, e^{2\pi \imath \l (\bbx_n - \bby_{n - 1})} \, K^*(\bx_n, \by_{n - 1}) \, \Delta(\by_{n - 1})
\end{equation}
and normalizing constant \eqref{dconst}. Taking the certain limit of the identity \eqref{Q*Lcomm} we obtain the following relation.
\begin{restatable}{thm}{LLthm}\label{LLthm}
	The operator identity
	\begin{equation}\label{LL*}
		\Lambda_{n}^*(\l) \, \Lambda_{n - 1}(\rho) = K_{2\hat{g}}(\l - \rho | \hat{\bo}) \,  \Lambda_n(\rho) \, \Lambda^*_{n - 1}(\l)
	\end{equation}
	holds true for $\l, \rho \in \mathbb{C}$ such that
	\begin{equation}
		| \Im(\l - \rho) | < \min(\nu_{g}, \nu_{g^*}).
	\end{equation}
\end{restatable}
Note that the coefficient in this relation explicitly reads
\begin{equation}
	K_{2\hat{g}}(\l - \rho | \hat{\bo}) = S^{-1}_2(\imath \l - \imath \rho + \hat{g}^*|\hat{\bo}) \, S^{-1}_2(\imath \rho - \imath \lambda + \hat{g}^*|\hat{\bo}).
\end{equation}
The proof of this exchange relation is given in Section \ref{sec:exch}. The whole chain of reductions looks as
\begin{align*}
&Q^*_n(\l)\, Q_n(\rho) = Q_n(\rho) \, Q_n^*(\l) &&\text{(Theorem \ref{theorem1})} \\[4pt]
& \hspace{2.25cm} \downarrow \\[4pt]
&Q^*_n(\l) \, \Lambda_n(\rho) = K_{\hat{g}}(\l - \rho | \hat{\bo}) \, \Lambda_n(\rho) \, Q^*_{n - 1}(\l) &&\text{(Theorem \ref{QLthm})}\\[4pt]
& \hspace{2.25cm} \downarrow \\[4pt]
&\Lambda_{n}^*(\l) \, \Lambda_{n - 1}(\rho) = K_{2\hat{g}}(\l - \rho | \hat{\bo}) \,  \Lambda_n(\rho) \, \Lambda^*_{n - 1}(\l) &&\text{(Theorem \ref{LLthm})}
\end{align*}
We also remark that similar limits performed among the original operators $Q_n(\l)$ and $\Lambda_n(\l)$ give the relations
\begin{align*}
	& Q_n(\l)\, Q_n(\rho) = Q_n(\rho) \, Q_n(\l) &&\text{\cite[Theorem 2]{BDKK}} \\[4pt]
	& \hspace{2.25cm} \downarrow \\[4pt]
	& Q_n(\l) \, \Lambda_n(\rho) = K_{\hat{g}^*}(\l - \rho | \hat{\bo}) \, \Lambda_n(\rho) \, Q_{n - 1}(\l) &&\text{\cite[Theorem 3]{BDKK2}}\\[4pt]
	& \hspace{2.25cm} \downarrow \\[4pt]
	& \Lambda_{n}(\l) \, \Lambda_{n - 1}(\rho) =  \Lambda_n(\rho) \, \Lambda_{n - 1}(\l) &&\text{\cite[Theorem 2]{BDKK2}}
\end{align*}

 Using the iterative representation of the wave function \eqref{I12} and the relation \eqref{Q*Lcomm} we arrive at the spectral description of the second Baxter $Q$-operator. Its proof with the necessary convergence arguments is given in Section~\ref{sec:Q*eigen}.
\begin{restatable}{thm}{QPsithm} \label{thm:QPsi}
	The wave function $\Psi_{ \bm{\lambda}_n }(\bm{x}_n)$ is a joint eigenfunction of the commuting family of operators $Q^*_{n}(\lambda)$
	\begin{equation}\label{QPsi}
		Q^*_{n}(\lambda) \, \Psi_{ \bm{\lambda}_n }(\bm{x}_n) = \prod_{j = 1}^n \hat{K}^*(\lambda-\lambda_j) \, \Psi_{ \bm{\lambda}_n }(\bm{x}_n).
	\end{equation}
	The integrals in both sides of \eqref{QPsi} converge if
	\begin{equation}
		| \Im(\l - \l_n) | < \frac{1}{2}(\nu_{g^*} - \epsilon \nu_g), \qquad |\Im(\l_k - \l_j)| \leq \theta(\epsilon), \qquad k,j = 1, \dots, n
	\end{equation}
	for any $\epsilon \in [0,1)$ and
	\begin{equation}
		\theta(\epsilon) = \frac{\nu_g}{2(n - 1)! e} \epsilon.
	\end{equation} 
\end{restatable}

On the other hand, 
the exchange relation \eqref{LL*} is crucial for the proof of the wave function symmetry with respect to the reflection $g \rightarrow g^*$, which is suggested by the formula \eqref{Msymm}. Introduce the counterpart of the function \eqref{eta}
\begin{equation}\label{coef}
 \hat{\eta}(\bl_n) := \prod_{\substack{i,j=1 \\ i\not=j}}^n S_2^{-1}(\imath \l_i - \imath \l_j + \hat{g}^* | \hat{\bo}).
\end{equation}

\begin{restatable}{thm}{Psisym} \label{thm:psi-g}
The wave function satisfies the relation
\begin{equation}\label{psi-g}
\Psi_{ \bm{\lambda}_n }(\bx_n; g | \bo) = \hat{\eta}^{-1}(\bl_n) \, \eta^{-1}(\bx_n) \, \Psi_{ \bm{\lambda}_n }(\bx_n; g^* | \bo).
\end{equation}
\end{restatable}
 
 The proof of Theorem \ref{thm:psi-g} based on the local relation \rf{LL*} is given in Section \ref{sec:eigen-g}. The symmetry \eqref{psi-g} was conjectured by M.~Halln\"as and S. Ruijsenaars in \cite{HR1} and was proven by them for $n=2$ in \cite{HR2} by different method. 

  The relation \rf{psi-g} complements the list of symmetries of the wave function $\Psi_{ \bm{\lambda}_n }(\bx_n; g | \bo)$ conjectured by S. Ruijsenaars and proven in \cite{BDKK2}: symmetry over spatial variables and over spectral variables, and the space-spectral duality
 \begin{equation}\label{duality}
	\Psi_{ \bm{\lambda}_n }(\bx_n; g | \bo) = \Psi_{ \bx_n }(\bl_n; \hat{g}^* | \hat{\bo}),
\end{equation}
see  \cite[Theorem 5]{BDKK2}. Note the consistency of the relations \rf{psi-g} and \rf{duality}.

In \cite{HR1} M.~Halln\"as and S. Ruijsenaars also conjectured that the poles of the wave function $\Psi_{\bl_n}(\bx_n; g |\bo)$ are located at the points
\begin{equation}\label{Psi-poles}
x_i - x_j = \imath g^* + m^1 \o_1 + m^2 \o_2, \qquad \l_i - \l_j = \imath \hat{g} + m^1 \hat{\o}_1 + m^2 \hat{\o}_2, \qquad m^k \geq 0
\end{equation}
for all $i, j$. This conjecture agrees with the formula \eqref{psi-g}: the points \eqref{Psi-poles} are precisely the poles of the coefficient $\hat{\eta}^{-1}(\bl_n) \, \eta^{-1}(\bx_n)$, see \eqref{A1a}.

Let us give one important corollary of both relations \eqref{psi-g}, \eqref{duality} and of the principal result of the paper \cite{HR3}. Introduce the function
\begin{equation}
\mu'(\bx_n) = \prod_{\substack{i, j = 1 \\ i < j}}^n \mu(x_i - x_j),
\end{equation}
so that for the measure function \eqref{I5} we have
\begin{equation}
\mu(\bx_n) = \mu'(\bx_n) \, \mu'(-\bx_n).
\end{equation}
Introduce also its counterpart with respect to the space-spectral duality \eqref{duality}
\begin{equation}
\hat{\mu}'(\bl_n) = \prod_{\substack{i, j = 1 \\ i < j}}^n \hat{\mu}(\l_i - \l_j), \qquad \hat{\mu}(\l) := \mu_{\hat{g}^*}(\l_i - \l_j | \hat{\bo}) = S_2(\imath \l | \hat{\bo}) S_2^{-1}(\imath \l + \hat{g}^*| \hat{\bo}).
\end{equation}
Finally, let us define a close relative of the wave function
\begin{equation}\label{Edef}
E_{\bl_n}(\bx_n) := E_{\bl_n}(\bx_n; g | \bo) =  e^{-\frac{\imath \hat{g}g^*}{4}n(n - 1)} \, \mu'(\bx_n) \, \hat{\mu}'(\bl_n) \, \Psi_{\bl_n}(\bx_n).
\end{equation}
In the paper \cite{HR3} M.~Halln\"as and S. Ruijsenaars obtained the asymptotics of this function with respect to $\l_j$ using the recursive representation \eqref{I12}. Namely, in the case of real periods $\o_1, \o_2$ they proved that 
\begin{equation}\label{Easymp1}
E_{\bl_n}(\bx_n) = \hat{E}^{\mathrm{as}}_{\bl_n}(\bx_n) + O \bigl( e^{-2 \pi r d(\bl_n)} \bigr), \qquad \l_j - \l_{j + 1} \to \infty
\end{equation}
with  $j = 1, \dots, n - 1$, where
\begin{equation}
d(\bl_n) = \min_{1 \leq i < j \leq n}(\l_i - \l_j) 
\end{equation}
and the asymptotic function is given by the formula
\begin{equation}
\hat{E}^{\mathrm{as}}_{\bl_n}(\bx_n) = \sum_{\sigma \in S_n} \prod_{\substack{i < j \\ \sigma^{-1}(i) > \sigma^{-1}(j)}} \, \frac{\mu(x_i - x_j)}{\mu(x_j - x_i)} \; \exp \biggl( 2\pi \imath \sum_{j = 1}^n \l_j x_{\sigma(j)} \biggr).
\end{equation}
It is also assumed that
\begin{equation}\label{r-g-cond}
r \in \Bigl[\frac{\min(\o_1, \o_2)}{2}, \min(\o_1, \o_2) \Bigl), \qquad \Re g \in \bigl(0, \max(\o_1, \o_2) \bigr].
\end{equation}
The duality relation \eqref{duality} gives the same type of asymptotics with respect to the coordinates $x_j$, and using the relation \eqref{psi-g} we extend the interval \eqref{r-g-cond} on $\Re g$.

\begin{restatable}{prop}{propAsymp}
	The function $E_{\bl_n}(\bx_n)$ has asymptotics
	\begin{equation}\label{Easymp}
	E_{\bl_n}(\bx_n) = E^{\mathrm{as}}_{\bl_n}(\bx_n) + O \bigl( e^{-2 \pi r d(\bx_n)} \bigr), \qquad x_j - x_{j + 1} \to \infty
	\end{equation}
	with $j = 1, \dots, n - 1$, where
	\begin{equation}\label{Eas}
	E^{\mathrm{as}}_{\bl_n}(\bx_n) = \sum_{\sigma \in S_n} \prod_{\substack{i < j \\ \sigma^{-1}(i) > \sigma^{-1}(j)}} \, \frac{\hat{\mu}(\l_i - \l_j)}{\hat{\mu}(\l_j - \l_i)} \; \exp \biggl( 2\pi \imath \sum_{j = 1}^n \l_{\sigma(j)} x_{j} \biggr)
	\end{equation}
	and we assume $\o_1, \o_2 > 0$ together with
	\begin{equation}
	r \in \Bigl[\frac{\min(\hat{\o}_1, \hat{\o}_2)}{2}, \min(\hat{\o}_1, \hat{\o}_2) \Bigl), \qquad \Re g \in (0, \o_1 + \o_2).
	\end{equation}
\end{restatable}

A short proof of this proposition is given in Section \ref{sec:eigen-g}. The coefficients behind the exponents in the asymptotic function \eqref{Eas} are the factorized scattering amplitudes of the Ruijsenaars hyperbolic system. Remarkably, they are connected to the $S$-matrices in the various field theories, see \cite{R4} and references therein.

It is a challenge to realize a wider group of symmetries, permuting parameters of eigenfunctions in accordance with Nekrasov connection to  fields theories \cite{N} and known symmetries of elliptic hypergeometric functions, see \cite{S}

\subsection{Relation to Noumi-Sano operators}
\def\bbe{\boldsymbol{e}}
Denote by $[x|\o_1]_m$ and $[x|\o_2]_m$ the following trigonometric Pochhammer symbols 
\begin{equation} \begin{split}\label{N0}
	[x|\o_1]_m =& \frac{S_2(x)}{S_2(x + m \o_1)} = \prod_{j = 0}^{m - 1} 2 \sin \frac{\pi (x + j \o_1)}{\o_2},\\
	[x|\o_2]_m =& \frac{S_2(x)}{S_2(x + m \o_2)} = \prod_{j = 0}^{m - 1} 2 \sin \frac{\pi (x + j \o_2)}{\o_1}.
\end{split}\end{equation}
Note that 
\begin{equation}\label{N1} [x+\o_1|\o_2]_m = (-1)^m[x|\o_2]_m. \end{equation} 
In \cite{NS} M. Noumi and A. Sano introduced an infinite family of difference operators,  that in the case of hyperbolic Ruijsenaars model have the following form 
\begin{equation}  \label{N2}
N_r^{(1)}(\bx_n) = (-1)^r \sum_{\substack{\bm{m}_n \in \mathbb{N}^n_0 \\[2pt] | \bm{m}_n | = r}} \, \prod_{i,j=1}^n \frac{[\imath x_i - \imath x_j + g|\o_1]_{m_i}}{[\imath x_i - \imath x_j - m_j \o_1|\o_1]_{m_i}} \prod_{i=1}^n T^{-\imath m_i \o_1}_{x_i}.
\end{equation} 
Here $\bm{m}_n = (m_1, \dots, m_n)$ is a sequence of non-negative integers and
\begin{equation} | \bm{m}_n | = m_1 + \ldots + m_n. \end{equation} 
Let us collect them into the generating series
\begin{equation} \label{N3} N^{(1)}(\bx_n;\l)=\sum_{r=0}^\infty (-1)^r e^{-2\pi\l r\o_1} \, N_r^{(1)}(\bx_n).\end{equation} 
Alternatively, we can consider Noumi-Sano operators with shifts by period $\o_2$
\begin{equation}  \label{N4}
N_r^{(2)}(\bx_n) = (-1)^r \sum_{\substack{\bm{m}_n \in \mathbb{N}^n_0 \\[2pt] | \bm{m}_n | = r}} \, \prod_{i,j=1}^n \frac{[\imath x_i - \imath x_j + g|\o_2]_{m_i}}{[\imath x_i - \imath x_j - m_j \o_2|\o_2]_{m_i}} \prod_{i=1}^n T^{-\imath m_i \o_2}_{x_i}
\end{equation}
collected into the generating series
\begin{equation} \label{N5} N^{(2)}(\bx_n;\l)=\sum_{r=0}^\infty (-1)^r e^{-2\pi\l r\o_2} \, N_r^{(2)}(\bx_n).\end{equation} 
As opposed to Macdonald operators, Noumi-Sano operators contain shifts by multiple periods $\o_i$. 
In the work \cite{NS} it is proven that these operators commute with Macdonald operators and between themselves
\begin{equation}\label{N3a}\begin{split} M_r(\bx_n;g|\bo) \, N_s^{(i)}(\bx_n) &=  N_s^{(i)}(\bx_n)\, M_n(\bx_n;g|\bo), \\[8pt] N_r^{(i)}(\bx_n)N_s^{(i)}(\bx_n)&=N_s^{(i)}(\bx_n)N_r^{(i)}(\bx_n)\end{split}\end{equation}
for any $r,s$ and $i=1,2$. 
Due to \rf{N1} Noumi-Sano operators of different kind also commute between themselves. Moreover, since they can be expressed via Macdonald operators by means of certain determinant formulas \cite[Proposition 1.3]{NS}, they also commute with both families of $Q$-operators.

In Section \ref{sec:NS} we observe certain relations between Noumi-Sano operators and the operator $Q^*_n(\l)$. On the one side, we can express the product of Noumi-Sano operators \rf{N3} and \rf{N4} of two kinds as certain contour integral with the kernel \rf{Q*ker} defining the operator
$Q^*_n(\l)$. 
Denote
\begin{equation}
\bbe_n= (1,\ldots,1) \in \C^n.
\end{equation}

\begin{restatable}{prop}{propNS} \label{propositionNS} 
	 For any non-negative $p,q \in \mathbb{Z}$  we have the equality of functions of $\bx_n$
	\begin{equation}
	\label{N8}\begin{split} 
		N_p^{(1)}(\bx_n) \, &N_q^{(2)}(\bx_n) \, f(\bx_n) =  (-1)^{p + q} \,  \biggl(\frac{2\pi S_2(g)}{\imath \sqrt{\o_1\o_2}}\biggr)^n \, e^{2\pi\l (p\o_1+q\o_2+\frac{n}{2}g )} \\[10pt]
		&\times \sum_{\substack{|\bm{m}|=p,\\[2pt]|\bk|=q}} \;
		\Res_{\substack{y_1=x_1-\imath(m_1\o_1+k_1\o_2)\\[2pt] \ldots \\[2pt] y_n=x_n-\imath(m_n\o_1+k_n\o_2)}} \,
		 Q^*\Bigl(\bx_n+\frac{\imath g}{2}\bbe_n,
		\by_n;\l \Bigr)f(\by_n) 
	\end{split}
	\end{equation}
	assuming $f(\bx_n)$ is analytical in the domain
	\begin{equation}\label{N11}
	-p \, \Re \o_1 - q \, \Re \omega_2 \leq \Im x_i \leq  0, \qquad i = 1, \dots, n.
	\end{equation}
\end{restatable}

Up to now, we do not have a satisfactory answer for the connection in opposite direction: how to express the operator $Q^*_n(\l)$ via Noumi-Sano operators. However, one can suggest the following partial result on a formal level. Namely let $f(\bx_n)$ be $\imath\o_2$-periodic symmetric function, such that  $Q_n^*(\l)f(\bx_n)$ may be computed by residues technique. It surely can happen only for periods with nonzero imaginary parts and for analytic function with not more than exponential growth. Then on this formal level we have the following proposition.
\begin{restatable}{prop}{propQNS}\label{propositionNS2}   
	 For a symmetric $\imath \o_2$-periodic analytic function $f(\bx_n)$ with not more than exponential growth  we have the equality
	\begin{equation}\label{N9}
	\bigl( Q^*_n(\l)  f \bigr) \Bigl(\bx_n+\frac{\imath g}{2}\bbe_n \Bigr) = e^{-\pi n g\l}c^{(2)}(\bx_n;\l) \, N^{(1)}(\bx_n;\l) \, f(\bx_n).
	\end{equation}
\end{restatable}
Here 
\begin{equation} \label{N10} c^{(2)}(\bx_n;\l)=N^{(2)}(\bx_n;\l) \, \mathbf{1} \end{equation}
is an application of the second Noumi-Sano operators generating function to the constant function equal to $1$. 

As we have learned after completing our work Hjalmar Rosengren presented similar ideas about $Q$-operators in the case of Ruijsenaars elliptic system in his talk during Elliptic Integrable Systems, Representation Theory and Hypergeometric Functions Workshop (July 2023).

\setcounter{equation}{0}

\section{Local relations}
\subsection{$Q^*Q$ commutativity} \label{sectionQQ*} 
\QQthm*

Both sides of the commutativity relation \rf{QQ*comm} are given by integral operators. Consider their kernels. Up to the same constant $d_n(g) d_n(g^*)$ from both sides the kernel of the left-hand side is
\begin{equation*}\begin{split}
	\eta^{-1}(\bx_n) \, \mu(\bz_n)\, \int_{\R^n} & d\by_n \,  e^{2\pi \imath \l(\bbx_n - \bby_n) + 2\pi \imath \rho(\bby_n - \bbz_n)}\prod_{\substack{i,j=1 \\ i\not=j}}^n  S_2(\imath (y_i -  y_j)) \, \\
	 &\times \prod_{i,j=1}^nS^{-1}_2 \Bigl(\pm\imath (x_i -  y_j) + \frac{g}{2}\Bigr) \,  S^{-1}_2 \Bigl(\pm\imath (z_i -  y_j) + \frac{g^*}{2}\Bigr)
	\end{split}\end{equation*}
and the kernel of the right-hand side is
\begin{equation*} \begin{split}
	\Delta(\bz_n) \, \int_{\R^n} & d\by_n \, e^{2\pi \imath \rho(\bbx_n - \bby_n) + 2\pi \imath \l(\bby_n - \bbz_n)} \prod_{\substack{i,j=1 \\ i\not=j}}^n  S_2(\imath (y_i -  y_j
	))  \\
	&  \times \prod_{i,j=1}^n S^{-1}_2 \Bigl(\pm\imath (x_i - y_j) + \frac{g^*}{2}\Bigr) \, S^{-1}_2\Bigl(\pm\imath (z_i -  y_j) + \frac{g}{2} \Big).
	\end{split} \end{equation*}
Here $\eta(\bx_n)$ is defined in \eqref{eta}, $\Delta(\bz_n)$ in \eqref{hvand} and we used compact notation
\begin{equation}
f(\pm z + c) = f(z + c) \, f(-z + c).
\end{equation} 
Due to the formula \eqref{mu-eta} the equality of kernels reduces to the following integral identity
 \begin{equation}\begin{split}
 	& \eta^{-1}(\bx_n) \, \eta(\bz_n) \int_{\R^n} d\by_n \,  e^{2\pi \imath (\rho-\l)\bby_n}\prod_{\substack{i,j=1 \\ i\not=j}}^n S_2(\imath (y_i -  y_j)) \\
 	& \hspace{1cm} \times \prod_{i,j=1}^nS^{-1}_2 \Bigl(\pm\imath (x_i -  y_j) + \frac{g}{2}\Bigr) \,  S^{-1}_2 \Bigl(\pm\imath (z_i -  y_j) + \frac{g^*}{2}\Bigr)\\[6pt]
 	& = e^{2\pi \imath (\rho-\l)(\bx_n+\bz_n)} \, \int_{\R^n} d\by_n  \, e^{2\pi \imath (\l-\rho)\bby_n} \prod_{\substack{i,j=1 \\ i\not=j}}^n S_2(\imath (y_i -  y_j))  \\
 	& \hspace{1cm} \times \prod_{i,j=1}^n S^{-1}_2 \Bigl(\pm\imath (x_i - y_j) + \frac{g^*}{2}\Bigr) \, S^{-1}_2\Bigl(\pm\imath (z_i -  y_j) + \frac{g}{2}\Bigr).
 	\end{split}\end{equation}
 	which is precisely the degeneration \rf{IM} of the Rains integral identity proven in Appendix \ref{AppB} under assumption \rf{lrho}.
 	Note that due to the reflection formula \rf{A3} the coefficient behind the integral equals
\begin{equation} \eta^{-1}(\bx_n) \, \eta(\bz_n)=\prod_{\substack{i, j =1 \\ i\not=j}}^n S^{-1}_2(\imath x_i - \imath x_j+g^\ast) \, S^{-1}_2(\imath z_i - \imath z_j+g).\end{equation}

\subsection{$Q^*\Lambda$ exchange relation}\label{sec:exch}
\QLthm*

The proof goes along the same lines, as the proof of exchange relation between the first $Q$-operator and $\Lambda$-operator, see \cite[Section 2]{BDKK2}.

\begin{proof}
Start from the commutativity relation
\begin{equation}\label{Q*Q2}
Q^*_n(\l) \, Q_n(\rho) = Q_n(\rho) \, Q^*_n(\l)
\end{equation}
written in terms of kernels
\begin{equation}\label{Q*Qker}
\int_{\R^n} d\by_n \, Q^*(\bx_n, \by_n;\l) \, Q(\by_n, \bz_n; \rho) = \int_{\R^n} d\by_n \, Q(\bx_n, \by_n;\rho) \, Q^*(\by_n, \bz_n; \l).
\end{equation}
The main idea of the proof is to take the limit $z_n \rightarrow \infty$ of this identity. To cancel asymptotic behavior of both sides with respect to $z_n$ we multiply the identity by the function
\begin{equation}
r(\bz_n; \rho) = \exp \Bigl( \pi \hat{g} \bigl[ \bbz_{n - 1} + (2 - n)z_n \bigr] + 2\pi \imath \rho z_n \Bigr)
\end{equation}
and then consider the limit. As we argued in Appendix \ref{AppB}, the integrals on both sides are absolutely convergent when
\begin{equation}\label{lr-cond}
| \Im (\l - \rho ) | < \frac{\nu_g + \nu_{g^*}}{2} = \Re \frac{\o_1 + \o_2}{\o_1 \o_2}.
\end{equation} 

The left-hand side of the identity \eqref{Q*Qker} multiplied by $r(\bz_n; \rho)$ equals to the integral
\begin{equation}\label{intF}
\int_{\R^n} d\bm{y}_n \, F(\bx_n, \by_n, \bz_n; \l, \rho)
\end{equation}
with
\begin{equation}\label{F}
\begin{aligned}
F & = e^{ 2\pi \imath \l ( \bbx_n - \bby_n ) + 2 \pi \imath \bigl( \rho - \frac{\imath \hat{g}}{2} \bigr) (\bby_n - \bbz_{n - 1} ) } \\[6pt]
& \times \eta^{-1}(\bx_n) \, K^*(\bx_n, \by_n) \, \Delta(\by_n) \, K(\by_n, \bz_{n - 1}) \, \mu(\bz_{n - 1}) \\[6pt]
& \times \prod_{j = 1}^n e^{ \pi \hat{g} (z_n - y_j) } K(z_n - y_j)  \prod_{j = 1}^{n - 1} e^{2 \pi \hat{g}(z_j - z_n)}  \mu(z_n - z_j)  \mu(z_j - z_n).
\end{aligned}
\end{equation}
The variable $z_n$ is contained only in the last line. Using asymptotics \eqref{Kmu-asymp}
\begin{equation}\label{Kmu-asymp2}
\mu(x) \sim e^{\pi \hat{g} | x | \pm \imath \frac{\pi \hat{g} g^*}{2} }, \qquad K(x) \sim e^{- \pi \hat{g} |x|}, \qquad x\rightarrow \pm \infty
\end{equation}
and bounds \eqref{Kmu-bound}
\begin{equation}\label{Kmu-bound2}
|\mu(x)| \leq C e^{\pi\nu_g |x|}, \qquad |K(x)| \leq C e^{-\pi\nu_g |x|},  \qquad x \in \R
\end{equation}
we deduce that the products in the last line \eqref{F} have pointwise limit
\begin{equation}
\lim_{z_n \rightarrow \infty} \, \prod_{j = 1}^n e^{ \pi \hat{g} (z_n - y_j) }  K(z_n - y_j)  \prod_{j = 1}^{n - 1} e^{2 \pi \hat{g}(z_j - z_n)}  \mu(z_n - z_j)  \mu(z_j - z_n) = 1
\end{equation}
and bounded independently of $z_n$
\begin{equation}\label{zn-bound}
\Biggl| \prod_{j = 1}^n e^{ \pi \hat{g} (z_n - y_j) }  K(z_n - y_j)  \prod_{j = 1}^{n - 1} e^{2 \pi \hat{g}(z_j - z_n)}  \mu(z_n - z_j)  \mu(z_j - z_n) \Biggr| \leq C_1(g, \bo),
\end{equation}
assuming sufficiently large $z_n$, such that $z_n > z_j$ for all $j = 1, \dots, n - 1$.

Next we bound the whole integrand $F$ with an integrable function and use dominated convergence theorem to take the limit of the integral. Recall that
\begin{equation}
\Delta(\by_n) = \prod_{\substack{i,j = 1 \\ i < j}}^n 4 \sh \frac{\pi(y_i - y_j)}{\o_1} \sh \frac{\pi(y_i - y_j)}{\o_2}.
\end{equation}
For hyperbolic sines in this product we have
\begin{equation}\label{sh-bound}
\biggl| \sh \frac{\pi(y_i - y_j)}{\o_a} \biggr| \leq \ch \biggl[ \Re \frac{\pi(y_i - y_j)}{\o_a} \biggr] \leq e^{ \pi \Re \frac{1}{\o_a} (|y_i| + |y_j|) }, \qquad a=1,2,
\end{equation}
where in the last step we used the assumption $\Re \o_a > 0$ and triangle inequality. Using again bound from \eqref{Kmu-bound2}, the analogous one with $g \rightarrow g^*$ and triangle inequalities we have
\begin{equation}\label{K*K-bound}
\begin{aligned}
| K^*(x_i - y_j) | &\leq C \, e^{- \pi \nu_{g^*} | x_i - y_j| } \leq C \, e^{\pi \nu_{g^*} (|x_i| - |y_j|)}, \\[6pt]
| K(z_i - y_j) | &\leq C \, e^{- \pi \nu_{g} | z_i - y_j| } \leq C \, e^{\pi \nu_{g} (|z_i| - |y_j|)}.
\end{aligned}
\end{equation}
Collecting \eqref{zn-bound}, \eqref{sh-bound} and \eqref{K*K-bound} we arrive at
\begin{equation}\label{Fbound}
|F| \leq C_2(g, \bo, \bx_n, \bz_{n-1}) \, \exp \Bigl( \bigl[ |2\pi \Im(\l - \rho) + \pi \nu_g | - \pi \nu_{g^*} \bigr] \| \by_n \| \Bigr)
\end{equation}
with some $C_2$, where by $\|\by_n\|$ we denote $L^1$-norm
\begin{equation}
\| \by_n \| = \sum_{j = 1}^n |y_j|.
\end{equation}
The bound \rf{Fbound} is an integrable function when
\begin{equation}\label{lr-cond2}
-\frac{\nu_{g^*} + \nu_g}{2} < \Im(\l - \rho) < \frac{\nu_{g^*} - \nu_g}{2}.
\end{equation}
This condition is stronger than the one we started with \eqref{lr-cond}. Assuming it we use dominated convergence theorem to write the limit of the integral \eqref{intF} as $z_n \rightarrow \infty$
\begin{equation}
\begin{aligned}
\int_{\R^n} & d\by_n \, e^{ 2\pi \imath \l ( \bbx_n - \bby_n ) + 2 \pi \imath \bigl( \rho - \frac{\imath \hat{g}}{2} \bigr) (\bby_n - \bbz_{n - 1} ) } \\[6pt]
& \times \eta^{-1}(\bx_n) \, K^*(\bx_n, \by_n) \, \Delta(\by_n) \, K(\by_n, \bz_{n - 1}) \, \mu(\bz_{n - 1}).
\end{aligned}
\end{equation}
This integral coincides with the kernel of the product $Q^*_n(\l) \, \Lambda_n(\rho - \imath \hat{g}/{2})$ up to constant $d_n(g^*) d_{n - 1}(g)$.

The integrand from the right-hand side of the kernels identity \eqref{Q*Qker} multiplied by $r(\bz_n;\rho)$ 
\begin{equation}\label{intG}
\int_{\R^n} d\by_n \, G(\bx_n, \by_n, \bz_n; \l, \rho)
\end{equation}
doesn't have pointwise limit as $z_n \rightarrow \infty$. To proceed we modify the integral in two steps. First, introduce the domain
\begin{equation}
D_j = \{ \by_n \in \R^n \colon y_j \geq y_k, \forall k \in [n] \setminus \{j\} \}.
\end{equation}
Here $[n] = \{1, \dots, n\}$. It is clear that
\begin{equation}\label{1+1}
\bm{1}_{D_1} + \bm{1}_{D_2} + \ldots + \bm{1}_{D_n} = 1
\end{equation}
where $\bm{1}_{D_j}$ is the indicator function of the domain $D_j$. The integrand $G$ in \eqref{intG} is symmetric with respect to $y_j$. Therefore using equality \eqref{1+1}
\begin{equation}
\int_{\R^n} d\by_n \, G(\bx_n, \by_n, \bz_n; \l, \rho) = n \int_{\R^n} d\by_n \, \bm{1}_{D_n} \, G(\bx_n, \by_n, \bz_n; \l, \rho).
\end{equation}
The second step is to shift the integration variable $y_n \rightarrow y_n + z_n$. The domain of indicator function after the shift changes to
\begin{equation}
D_n' = \{ \by_n \in \R^n\colon y_n + z_n \geq y_k, \forall k \in [n - 1] \}.
\end{equation}
The whole integrand after the shift
\begin{equation}\label{G}
\begin{aligned}
&\bm{1}_{D_n'} \, G(\bx_n, \by_{n - 1}, y_n + z_n, \bz_n; \l,\rho) = e^{2\pi \imath \bigl( \rho - \frac{\imath \hat{g}}{2} \bigr) (\bbx_n - \bby_n) + 2\pi \imath \lambda (\bby_n - \bbz_{n - 1})} \\[8pt]
&\times K(\bx_n, \by_{n - 1}) \, \Delta(\by_{n - 1}) \, K^*(\by_{n - 1}, \bz_{n - 1}) \, K^*(y_n) \, \Delta(\bz_{n - 1}) \, R(\bx_n, \bz_n, \by_n)
\end{aligned}
\end{equation}
where the variable $z_n$ is contained only in the last function
\begin{equation}
\begin{aligned}
R(\bx_n, &\by_n, \bz_n) =  \prod_{j = 1}^{n - 1} e^{\pi \hat{g}^*(y_{nj} + z_{nj} + z_n)} K^*(y_n + z_{nj}) K^*(y_j - z_n) \\[6pt] 
&\times \prod_{j = 1}^n  e^{\pi \hat{g} (y_n + z_n - x_j)} K(x_n - y_n - z_n )  \, \prod_{j = 1}^{n - 1} e^{\pi (\hat{g} + \hat{g}^*)z_{jn}} \, 4 \sh \frac{\pi z_{nj}}{\o_1} \sh \frac{\pi z_{nj}}{\o_2} \\[8pt]
&\times  \prod_{j = 1}^{n - 1} e^{\pi (\hat{g} + \hat{g}^*)(y_{jn} - z_n)} 4 \sh \frac{\pi(y_{nj} + z_n)}{\o_1} \sh \frac{\pi(y_{nj} + z_n)}{\o_2} \cdot \bm{1}_{D_n'}.
\end{aligned}
\end{equation}
Here, for brevity, we used notation $y_{jn} = y_j - y_n$. Due to the asymptotics \eqref{Kmu-asymp2} the last function has pointwise limit
\begin{equation}
\lim_{z_n \rightarrow \infty} R(\bx_n, \by_n, \bz_n) = 1.
\end{equation}
Also note that in the presence of the indicator function
\begin{equation}
y_{nj} + z_n = | y_{nj} + z_n |, \qquad j = 1, \dots, n - 1.
\end{equation}
Using it together with the fact that $| \sh( z ) | \leq 2 e^{|\Re z|}$ we bound factors from the last product in $R$ 
\begin{equation}
\left| e^{\pi (\hat{g} + \hat{g}^*)(y_{jn} - z_n)} 4 \sh \frac{\pi(y_{nj} + z_n)}{\o_1} \sh \frac{\pi(y_{nj} + z_n)}{\o_2} \right| \leq 16.
\end{equation}
Factors from three other products are estimated similarly using \eqref{Kmu-bound2}, so that for $R$ we have the bound independent of $z_n$
\begin{equation}
| R(\bx_n, \by_n, \bz_n) | \leq C_3(g, \bo).
\end{equation}

Using the last bound together with \eqref{K*K-bound} we estimate the integrand \eqref{G}
\begin{equation}
\begin{aligned}
| \bm{1}_{D'_n} \,& G(\bx_n, \by_{n - 1}, y_n + z_n, \bz_n; \l,\rho) |  \\[6pt]
& \leq C_4 \exp \Bigl( \bigl[ | 2\pi \Im(\rho - \l) - \pi \nu_g | - 2\pi \nu_g - \pi \nu_{g^*} \bigr] \| \by_{n - 1} \| \\[6pt]
& \hspace{1.47cm} + \bigl[ |2\pi \Im(\rho - \l) - \pi \nu_g | - \pi \nu_{g^*} \bigr] | y_n| \Bigr)
\end{aligned}
\end{equation}
with some $C_4(g, \bo, \bx_n, \bz_{n - 1})$. Function from the right is integrable under the same condition which appeared in the limit of the left-hand side \eqref{lr-cond2}. Assuming it we may use dominated convergence theorem. The limit of the right-hand side integral \eqref{intG} equals
\begin{equation}
\begin{aligned}
n \int_{\R^n} & d\by_n \, e^{2\pi \imath \bigl( \rho - \frac{\imath \hat{g}}{2} \bigr) (\bbx_n - \bby_n) + 2\pi \imath \lambda (\bby_n - \bbz_{n - 1})} \\[8pt]
&\times K(\bx_n, \by_{n - 1}) \, \Delta(\by_{n - 1}) \, K^*(\by_{n - 1}, \bz_{n - 1}) \, K^*(y_n) \, \Delta(\bz_{n - 1}).
\end{aligned}
\end{equation}
The integral over $y_n$ has separated. It is just a Fourier transform of the function $K^*$, which differs from the Fourier of $K$ by the exchange $g \rightarrow g^*$, so by \eqref{K-fourier},
\begin{equation}
\int_\R dy_n \, e^{2 \pi \imath \bigl( \l - \rho + \frac{\imath \hat{g}}{2} \bigr) y_n} \, K^*(y_n) = \sqrt{\o_1 \o_2} \, S_2(g^*) \, \hat{K}^* \Bigl( \l - \rho + \frac{\imath \hat{g}}{2} \Bigr).
\end{equation} 
The remaining part of the integral coincides with the kernel of the product of operators $\Lambda_n(\rho - \imath \hat{g}/2) \, Q_{n - 1}^*(\l)$ up to constant $d_{n - 1}(g) d_{n - 1}(g^*)$.

Thus, taking the limit of the commutativity relation \eqref{Q*Q2} we arrive at the identity
\begin{equation}
Q^*_n(\l) \, \Lambda_n \Bigl( \rho - \frac{\imath \hat{g} }{2} \Bigr) = \hat{K}^*\Bigl( \l - \rho + \frac{\imath \hat{g} }{2} \Bigr) \, \Lambda_n \Bigl( \rho - \frac{\imath \hat{g} }{2} \Bigr) \, Q^*_{n - 1} (\l)
\end{equation}
where we assume
\begin{equation}
-\frac{\nu_{g^*} + \nu_g}{2} < \Im(\l - \rho) < \frac{\nu_{g^*} - \nu_g}{2}.
\end{equation}
The identity stated in the theorem follows by the shift of the parameter $\rho \rightarrow \rho + \imath \hat{g}/2$.
\end{proof}

\subsection{$\Lambda^*\Lambda$ exchange relation}
\LLthm*

\begin{proof}
This proof is very similar to the previous one. Let us write the exchange relation
\begin{equation}
Q^*_n(\l) \, \Lambda_n(\rho) = \hat{K}^*(\l - \rho) \, \Lambda_n(\rho) \, Q^*_{n - 1}(\l)
\end{equation}
in terms of kernels \eqref{Q*ker}, \eqref{Lker}
\begin{equation}\label{Q*Lker}
\begin{aligned}
\int_{\R^{n}}& d\by_{n} \, Q^*(\bx_n, \by_{n}; \l) \, \Lambda(\by_{n}, \bz_{n - 1}; \rho) \\[6pt]
&= n \sqrt{\o_1 \o_2} \, S_2(g^*) \, \hat{K}^*(\l - \rho) \, \int_{\R^{n - 1}} d\by_{n - 1} \, \Lambda(\bx_{n}, \by_{n - 1}; \rho) \, Q^*(\by_{n - 1}, \bz_{n - 1}; \l).
\end{aligned}
\end{equation}
As it is shown during the previous proof, this identity holds under assumption
\begin{equation}\label{Q*Lker-cond}
| \Im(\l - \rho) | < \frac{\nu_{g^*}}{2}.
\end{equation}
The idea is to take its limit as $z_{n - 1} \rightarrow \infty$. In this limit the operators $Q^*_k(\l)$ on both sides turn into the operators $\Lambda^*_k(\l)$ with the kernels \eqref{L*ker}. To cancel asymptotic behavior of both sides with respect to $z_{n - 1}$ we multiply them by the function
\begin{equation}
r(\bz_{n - 1}; \l) = \exp \Bigl( \pi \hat{g} \, \bz_{n - 2} + \pi \bigl[ \hat{g}^* + (2 - n)\hat{g} + 2\imath \l \bigr] z_{n - 1} \Bigr).
\end{equation}

Consider the right-hand side of the identity \eqref{Q*Lker} multiplied by $r(\bz_{n - 1}; \l)$. Without the coefficient behind the integral it equals
\begin{equation}
\int_{\R^{n - 1}} d\by_{n - 1} \, F(\bx_n, \by_{n - 1}, \bz_{n - 1}; \l, \rho)
\end{equation}
where
\begin{equation}
\begin{aligned}
F &= e^{2\pi \imath \rho ( \bbx_n - \bby_{n - 1} ) + 2\pi \imath \bigl( \l - \frac{\imath \hat{g}^*}{2} \bigr) (\bby_{n - 1} - \bbz_{n - 2} ) } \\[6pt]
&\times K(\bx_n, \by_{n - 1}) \, \Delta(\by_{n - 1}) \, K^*(\by_{n - 1}, \bz_{n - 2}) \, \Delta(\bz_{n - 2}) \\[6pt]
&\times \prod_{j = 1}^{n - 1} e^{\pi \hat{g}^* (z_{n - 1} - y_j)} K^*(z_{n - 1} - y_j) \prod_{j = 1}^{n - 2} e^{\pi (\hat{g} + \hat{g}^*) z_{j, n - 1} } 4 \sh \frac{\pi z_{j, n - 1}}{\o_1} \sh \frac{\pi z_{j, n - 1}}{\o_2}.
\end{aligned}
\end{equation}
Here we denoted $z_{j, n - 1} = z_j - z_{n - 1}$. Using asymptotics \eqref{Kmu-asymp2} we deduce that the last line tends to $1$ as $z_{n - 1} \rightarrow \infty$. Then using bounds \eqref{sh-bound}, \eqref{K*K-bound} we derive estimate for the whole integrand
\begin{equation}
| F | \leq C_1(g, \bo, \bx_n, \bz_{n - 2}) \, \exp \Bigl( \bigl[ | 2\pi \Im(\l - \rho) - \pi \nu_{g^*} | - 2 \pi \nu_g  \bigr] \| \by_{n - 1} \| \Bigr) 
\end{equation}
independent of $z_{n - 1}$ (for sufficiently large $z_{n - 1}$ such that $z_{n - 1} > z_j$, $j = 1, \dots,n - 2$). The function from the right is integrable under assumption
\begin{equation}
\frac{\nu_{g^*}}{2} - \nu_g <  \Im(\l - \rho) <  \frac{\nu_{g^*}}{2} + \nu_g.
\end{equation}
The overlap between this assumption and initial one \eqref{Q*Lker-cond} is as follows
\begin{equation}\label{rhs-cond}
\max \Bigl( \frac{\nu_{g^*}}{2} - \nu_g, - \frac{\nu_{g^*}}{2} \Bigr) < \Im(\l - \rho) < \frac{\nu_{g^*}}{2}.
\end{equation} 
Under this condition we use dominated convergence theorem and in the limit obtain
\begin{equation}
\begin{aligned}
\int_{\R^{n - 1}} &d\by_{n - 1} \, e^{2\pi \imath \rho ( \bbx_n - \bby_{n - 1} ) + 2\pi \imath \bigl( \l - \frac{\imath \hat{g}^*}{2} \bigr) (\bby_{n - 1} - \bbz_{n - 2} ) } \\[6pt]
&\times K(\bx_n, \by_{n - 1}) \, \Delta(\by_{n - 1}) \, K^*(\by_{n - 1}, \bz_{n - 2}) \, \Delta(\bz_{n - 2}).
\end{aligned}
\end{equation}
This integral coincides with the kernel of $\Lambda_n(\rho) \, \Lambda^*_{n - 1}(\l - \imath \hat{g}^*/2)$ up to the constant $d_{n - 1}(g) \, d_{n - 2}(g^*)$.

To take the limit of the left-hand side of the relation \eqref{Q*Lker} multiplied by the function $r(\bz_{n - 1}; \l)$
\begin{equation}
\int_{\R^n} d\by_n \, G(\bx_n, \by_n, \bz_{n - 1}; \l, \rho)
\end{equation} 
we use the same trick, as in the previous proof. Firstly, we rewrite the integral using symmetry of the integrand $G$ with respect to $y_j$ and the identity \eqref{1+1}
\begin{equation}
\int_{\R^n} d\by_n \, G(\bx_n, \by_n, \bz_{n - 1}; \l, \rho) = n \int_{\R^n} d\by_n \, \bm{1}_{D_n} G(\bx_n, \by_n, \bz_{n - 1}; \l, \rho),
\end{equation}
where the domain of the indicator function
\begin{equation}
D_n = \{ \by_n \in \R^n\colon y_n \geq y_k, \forall k \in [n - 1] \}.
\end{equation}
Secondly, we shift the integration variable $y_n \rightarrow y_n + z_{n - 1}$, so that we have the integrand
\begin{equation}
\begin{aligned}
\bm{1}_{D_n'} \, G(\bx_n, \by_{n - 1}, & y_n + z_{n - 1},  \bz_{n - 1}) = e^{2\pi \imath \bigl( \l - \frac{\imath \hat{g}^*}{2} \bigr)(\bbx_n - \bby_{n - 1}) + 2\pi \imath \rho (\bby_{n - 1} - \bbz_{n - 2})} \\[6pt] 
&\times e^{ 2\pi \imath \bigl( \rho - \l + \frac{\imath(\hat{g}^* - \hat{g})}{2} \bigr)y_{n} } \, \eta^{-1}(\bx_n) \, K^*(\bx_n, \by_{n - 1}) \, \Delta(\by_{n - 1})  \\[6pt]
&\times K(\by_{n - 1}, \bz_{n - 2}) \, K(y_n) \, \mu(\bz_{n - 2}) \, R(\bx_n, \by_n, \bz_{n - 1})
\end{aligned}
\end{equation}
where the variable $z_{n - 1}$ is contained only in the last function
\begin{equation}
\begin{aligned}
R &= \prod_{j = 1}^n e^{ \pi \hat{g}^* (y_n + z_{n - 1} - x_j) } K^*(y_n + z_{n - 1} - x_j) \prod_{j = 1}^{n - 2} e^{\pi \hat{g}(y_n + z_{n - 1,j}) } K(y_n + z_{n - 1,j})  \\[6pt]
&\times \prod_{j = 1}^{n - 1} e^{\pi \hat{g}(z_{n - 1} - y_j) } K(z_{n - 1} - y_j) \prod_{j = 1}^{n - 2} e^{2\pi \hat{g} z_{j, n - 1} }\mu(z_{n - 1, j}) \mu(z_{j, n - 1}) \\[6pt] 
&\times \prod_{j = 1}^{n - 1} e^{\pi(\hat{g} + \hat{g}^*) (y_{jn}-z_{n - 1}) } 4\sh \frac{\pi(y_{jn} - z_{n - 1})}{\o_1} \sh \frac{\pi(y_{jn} - z_{n - 1})}{\o_2} \cdot \bm{1}_{D_n'}.
\end{aligned}
\end{equation}
Note that the domain of the indicator function changed after the shift
\begin{equation}
D_n' = \{ \by_n \in \R^n\colon y_n + z_{n - 1} \geq y_k, \forall k \in [n - 1] \}.
\end{equation}
Using again the same asymptotics \eqref{Kmu-asymp2} and bounds \eqref{Kmu-bound2}, \eqref{sh-bound} we deduce that 
\begin{equation}
\lim_{z_{n - 1} \rightarrow \infty} R \, = \, 1, \hspace{2cm} |R| \leq C_2(g, \bo)
\end{equation}
assuming sufficiently large $z_{n - 1}$. Also from the same bounds we derive the estimate for the whole integrand
\begin{equation}
\begin{aligned}
\bigl| \bm{1}_{D_n'} \, &G(\bx_n, \by_{n - 1},  y_n + z_{n - 1},  \bz_{n - 1}) \bigr| \\[6pt]
&\leq C_3 \exp \Bigl( \bigl[ |2\pi \Im(\rho - \l) + \pi \nu_{g^*} | - 2 \pi \nu_{g^*} \bigr] \| \by_{n - 1} \| \\[6pt]
& \hspace{1.45cm} + \bigl[ |2\pi\Im(\rho - \l) + \pi \nu_{g^*} - \pi \nu_{g} | - \pi \nu_g \bigr] |y_n| \Bigr).
\end{aligned}
\end{equation}
The function from the right is integrable in some strip of $\Im(\l - \rho)$, whose intersection with initial strip \eqref{Q*Lker-cond} gives the same condition that appeared for the right-hand side \eqref{rhs-cond}. Assuming it we again use dominated convergence theorem and in the limit obtain
\begin{equation}
\begin{aligned}
n \int_{\R^n} &d\by_n \, e^{2\pi \imath \bigl( \l - \frac{\imath \hat{g}^*}{2} \bigr)(\bbx_n - \bby_{n - 1}) + 2\pi \imath \rho (\bby_{n - 1} - \bbz_{n - 2}) + 2\pi \imath \bigl( \rho - \l + \frac{\imath(\hat{g}^* - \hat{g})}{2} \bigr)y_{n}} \\[6pt]
&\times  \eta^{-1}(\bx_n) \, K^*(\bx_n, \by_{n - 1}) \, \Delta(\by_{n - 1}) \, K(\by_{n - 1}, \bz_{n - 2})  \, \mu(\bz_{n - 2}) \, K(y_n).
\end{aligned}
\end{equation}
Note that the integral over $y_n$ has separated and represents a Fourier transform of the function $K$ \eqref{K-fourier}
\begin{equation}
\int_\R dy_n \, e^{2\pi \imath \bigl( \rho - \l + \frac{\imath(\hat{g}^* - \hat{g})}{2} \bigr)y_{n}} \, K(y_n) = \sqrt{\o_1 \o_2} \, S_2(g) \, \hat{K} \Bigl( \rho - \l + \frac{\imath(\hat{g}^* - \hat{g})}{2} \Bigr).
\end{equation}
The rest part of the integral coincides with the kernel of $\Lambda_n^*(\l - \imath \hat{g}^*/2) \, \Lambda_{n - 1}(\rho)$ up to constant $d_{n - 1}(g^*) \, d_{n - 2}(g)$. 

Collecting everything together, in the limit we have the identity
\begin{equation}\label{L*L-shift}
\begin{aligned}
\hat{K} \Bigl( \rho - \l + \frac{\imath(\hat{g}^* - \hat{g})}{2} \Bigr) \, & \Lambda^*_n \Bigl( \l - \frac{\imath \hat{g}^*}{2} \Bigr) \, \Lambda_{n - 1}(\rho) \\[6pt]
&= \hat{K}^*(\l - \rho) \, \Lambda_{n }(\rho) \, \Lambda^*_{n - 1} \Bigl( \l - \frac{\imath \hat{g}^*}{2} \Bigr)
\end{aligned}
\end{equation}
that holds true assuming the condition \eqref{rhs-cond}. Note that coefficients in this identity have common factor
\begin{align}
&\hat{K}^*(\l - \rho) = S_2^{-1} \Bigl( \imath \l - \imath \rho + \frac{\hat{g}^*}{2} \Big| \hat{\bo} \Bigr) \, S_2^{-1} \Bigl( \imath \rho - \imath \l + \frac{\hat{g}^*}{2} \Big| \hat{\bo} \Bigr), \\[6pt]
&\hat{K} \Bigl( \rho - \l + \frac{\imath(\hat{g}^* - \hat{g})}{2} \Bigr) = S_2^{-1} \Bigl( \imath \rho - \imath \l + \hat{g} - \frac{\hat{g}^*}{2} \Big| \hat{\bo} \Bigr) \, S_2^{-1} \Bigl( \imath \l - \imath \rho + \frac{\hat{g}^*}{2} \Big| \hat{\bo} \Bigr).
\end{align}
Canceling it from both sides, shifting $\l \rightarrow \l + \imath \hat{g}^*/2$ and using reflection formula \eqref{A3} we arrive at
\begin{equation}\label{L*L2}
\Lambda_n^*(\l) \, \Lambda_{n - 1}(\rho) = S_2^{-1} ( \imath \l - \imath \rho + \hat{g}^* | \hat{\bo} ) \, S_2^{-1} ( \imath \rho - \imath \l + \hat{g}^* | \hat{\bo} ) \, \Lambda_n(\rho) \, \Lambda^*_{n - 1}(\l)
\end{equation}
with the condition
\begin{equation}
- \min(\nu_g, \nu_{g^*}) < \Im(\l - \rho) < 0.
\end{equation}
From the estimates obtained in the proof it is clear that the left-hand side of the final relation is analytic in a wider strip $| \Im(\l - \rho) | < \nu_{g^*}$ and similarly the right-hand side is analytic in the strip $| \Im(\l - \rho) | < \nu_g$. Therefore, we can analytically continue this relation to the strip
\begin{equation}
| \Im(\l - \rho) | < \min(\nu_g, \nu_{g^*}).
\end{equation}
\end{proof}

\begin{remark*} 
The local relation \eqref{L*L2} can be also obtained directly from the Rains integral identity \eqref{id1} similarly to the commutativity relation \eqref{QQ*comm}.
\end{remark*}

\setcounter{equation}{0} 
\section{Eigenfunctions}
\subsection{Eigenfunctions of $Q^*$-operator}\label{sec:Q*eigen} 
\QPsithm*

The proof is close to the one given in our previous paper \cite[Section 3]{BDKK2} about diagonalization of the first $Q$-operator. We use recursive structure of the wave function 
\begin{equation}
\Psi_{\bl_n}(\bx_n) = \Lambda_n(\lambda_n) \, \Psi_{\bl_{n - 1}}(\bx_{n - 1})
\end{equation}
together with the exchange relation given by Theorem 2
\begin{equation}\label{Q*L}
Q^*_n(\l) \, \Lambda_n(\rho) = \hat{K}^*(\l - \rho) \, \Lambda_n(\rho)  \, Q^*_{n - 1}(\l).
\end{equation}
The subtle point is the convergence of multiple integrals \eqref{Iint}, \eqref{Jint} appearing during this calculation.

\begin{restatable}{prop}{prop1} \label{prop1} \-
	\begin{itemize}
		\item[(i)] The multiple integral
		\begin{equation}\label{Iint}
		I_{\l, \bl_n} = Q^*_n(\l) \, \Lambda_n(\l_n) \, \cdots \, \Lambda_2(\l_2) \, e^{2\pi \imath \l_1 x_1} 
		\end{equation}
		is absolutely convergent in the domain
		\begin{equation}\label{I-dom}
		| \Im(\l - \l_n) | < \frac{1}{2}(\nu_{g^*} - \epsilon \nu_g), \qquad |\Im(\l_k - \l_j)| \leq \theta(\epsilon), \qquad k,j = 1, \dots, n
		\end{equation}
		for any $\epsilon \in [0,1)$ and
		\begin{equation}
		\theta(\epsilon) = \frac{\nu_g}{2(n - 1)! e} \epsilon.
		\end{equation} 
		Moreover, it is analytic with respect to $\l, \bl_n$ on compact subsets of this domain.
		\item[(ii)] The multiple integral 
		\begin{equation}\label{Jint}
		J_{\l, \bl_n} = \Lambda_n(\l_n) \, Q^*_{n - 1}(\l) \, \Lambda_{n - 1}(\l_{n - 1}) \, \cdots \, \Lambda_2(\l_2) \, e^{2\pi \imath \l_1 x_1} 
		\end{equation}
		is absolutely convergent under the restriction
		\begin{equation}
		\Im\l = \Im \l_k, \qquad  k = 1, \dots, n.
		\end{equation}
	\end{itemize}
\end{restatable}
\begin{proof}[Proof of Proposition \ref{prop1}]
	In our previous work we already considered analogous integrals, but with the operators $Q_k(\l)$ instead of $Q^*_k(\l)$ \cite[Proposition 1]{BDKK2}. The convergence of the present integrals follows from almost the same bounds and inequalities.
	
	Consider the first integral \eqref{Iint}. It has $n$ groups of integration variables, denote them~by
	\begin{equation}
	\by_k = \bigl( y_1^{(k)}, \dots, y_k^{(k)} \bigr).
	\end{equation}
	The integral in full form
	\begin{equation}
	\begin{aligned}
	I_{\l, \bl_n} = C_I \, \eta^{-1}(\bx_n) \int & \prod_{j = 1}^n d\by_j  \; \Delta(\by_n) \, K^*(\bx_n, \by_n) \, e^{2\pi \imath \bigl[ \l \bbx_n + (\l_n - \l) \bby_n \bigr]} \\[6pt]
	& \times \prod_{j = 1}^{n - 1} \mu(\by_j) \, K(\by_{j + 1}, \by_j) \, e^{2\pi \imath (\l_j - \l_{j + 1}) \bby_j},
	\end{aligned}
	\end{equation}
	where in the last product we assume $\mu(\by_1) \equiv 1$. The constant $C_I$ contains all constants $d_k$ from operators and all integrals are over $\R$. Denote the integrand by $F$ and suppose
	\begin{equation}
	| \Im(\l - \l_n) | \leq \delta_Q \frac{\nu_{g^*}}{2} , \qquad |\Im(\l_{j} - \l_{j + 1}) | \leq \delta_\Lambda \frac{\nu_g}{2}, \qquad j = 1, \dots, n - 1.
	\end{equation}
	Note also that from definition \eqref{hvand} we have
	\begin{equation}\label{hvand-b}
	\begin{aligned}
	| \Delta( \by_n) | & = \Biggl|\prod_{\substack{i,j = 1 \\ i < j}}^n 4 \sh \frac{\pi(y_i - y_j)}{\o_1} \sh \frac{\pi(y_i - y_j)}{\o_2} \Biggr| \\
	& \leq C \, \exp \pi \biggl( \frac{\nu_g + \nu_{g^*}}{2} \sum_{ \substack{ i, j = 1 \\ i \not= j} }^n \bigl| y_i^{(n)} - y_j^{(n)} \bigr| \biggr).
	\end{aligned}
	\end{equation}
	Then using bounds \eqref{hvand-b}, \eqref{Kmu-bound} and triangle inequalities we obtain the estimate 
	\begin{equation}\label{I-int}
	\begin{aligned}
	|F| \leq C_1 \exp \pi \biggl( \bigl[\delta_Q - n\bigr] \nu_{g^*} \| \by_n \| &+ \frac{\nu_{g^*} - \nu_g}{2} \sum_{ \substack{ i, j = 1 \\ i \not= j} }^n \bigl| y_i^{(n)} - y_j^{(n)} \bigr| \\
	&+ \nu_g S_n(\by_1, \dots, \by_n) + \delta_\Lambda \nu_g \sum_{k = 1}^{n - 1} \| \by_k \| \biggr)
	\end{aligned}
	\end{equation}
	with some $C_1(g, \bo, \bx_n)$ and the function $S_n$ defined by recurrence relation
	\begin{equation}\label{Srec}
	\begin{aligned}
	S_n (\bm{y}_1, \dots, \bm{y}_n ) = \sum_{\substack{i, j = 1 \\ i \not= j}}^n \bigl| y_i^{(n)} - y_j^{(n)} \bigr| &- \sum_{i = 1}^{n} \sum_{j = 1}^{n - 1} \, \bigl| y_i^{(n)} - y_j^{(n - 1)} \bigr| \\
	&+ S_{n - 1} (\bm{y}_1, \dots,  \bm{y}_{n - 1} )
	\end{aligned}
	\end{equation}
	with $S_1 = 0$. In the previous paper we proved the following bound on $S_n$ \cite[Lemma~2]{BDKK2}
	\begin{equation}\label{Sn-bound}
	S_n \leq \frac{1}{2} \sum_{\substack{i, j = 1 \\ i \not= j}}^n \bigl| y_i^{(n)} - y_j^{(n)} \bigr| + \ve \| \by_n \| - \frac{\ve}{c_n} \, \sum_{k = 1}^{n - 1} \| \bm{y}_k \|
	\end{equation}
	for any $\epsilon \in \left[ 0,2(n - 1) \right]$, where the numbers $c_n$ are bounded as 
	\begin{equation}
	c_n < (n - 1)! e. 
	\end{equation}
	Substituting it into \eqref{I-int} and using the estimate 
	\begin{equation}\label{sum<}
	\frac{1}{2} \sum_{ \substack{ i, j = 1 \\ i \not= j} }^n \bigl| y_i^{(n)} - y_j^{(n)} \bigr| \leq (n - 1) \| \by_{n} \|
	\end{equation}
	we have
	\begin{equation}
	|F| \leq C_1 \exp \pi \biggl( \bigl[ \delta_Q \, \nu_{g^*} - \nu_{g^*} + \ve \nu_g \bigr]  \| \by_n \| + \biggl[ \delta_\Lambda - \frac{\ve}{c_n} \biggr] \nu_g \sum_{k = 1}^{n - 1} \| \by_k \| \biggr).
	\end{equation}
	So, the function from the right is integrable under assumptions
	\begin{equation}
	\delta_Q < 1 - \frac{ \ve \nu_g}{\nu_{g^*}}, \qquad \delta_\Lambda \leq \frac{\ve}{(n - 1)! e}.
	\end{equation}
	
	Next consider the integral \eqref{Jint}. In full form it looks as follows
	\begin{equation}
	\begin{aligned}
	&J_{\l, \bl_n}  = C_J \int d\bt_{n - 1} \prod_{j = 1}^{n - 1} d\by_j \; \Delta(\bt_{n - 1}) \, K(\bx_n, \bt_{n - 1}) \, e^{2 \pi \imath \bigl[ \l_n \bbx_n + (\l - \l_n) \bbt_{n - 1} \bigr]} \times\\
	&  \Delta(\by_{n - 1}) \, K^*(\bt_{n - 1}, \by_{n - 1}) \, e^{2\pi \imath ( \l_{n - 1} - \l) \bby_{n - 1} } \prod_{j = 1}^{n - 2} \mu(\by_j) \, K(\by_{j + 1}, \by_j) \, e^{2\pi \imath (\l_j - \l_{j+1}) \bby_j}.
	\end{aligned}
	\end{equation}
	Denote the integrand by $G$. Assuming
	\begin{equation}
	\Im \l = \Im \l_j
	\end{equation}
	for all $j$ use bounds \eqref{hvand-b}, \eqref{Kmu-bound} and triangle inequalities to arrive at
	\begin{equation}
	\begin{aligned}
	|G| \leq C_2 \exp \pi \biggl(& - n \nu_g \| \bt_{n - 1} \| + \nu_g \sum_{ \substack{ i,j = 1 \\ i < j } }^{n - 1} \Bigl( \bigl|t_i - t_j \bigr| + \bigl| y_i^{(n - 1)} - y_j^{(n - 1)} \bigr| \Bigr) \\
	&+ \nu_{g^*} R_{n - 1}(\bt_{n - 1}, \by_{n - 1}) + \nu_g S_{n - 1}(\by_1, \dots, \by_{n - 1}) \biggr),
	\end{aligned}
	\end{equation}
	where we introduced new function
	\begin{equation}
	R_{n - 1} = \sum_{ \substack{ i,j = 1 \\ i < j } }^{n - 1} \Bigl( \bigl|t_i - t_j \bigr| + \bigl| y_i^{(n - 1)} - y_j^{(n - 1)} \bigr| \Bigr) - \sum_{i,j = 1}^{n - 1} \bigl| t_i - y_j^{(n - 1)}\bigr|.
	\end{equation}
	In our previous paper we proved the bound \cite[Corollary 1]{BDKK3}
	\begin{equation}
	R_{n - 1}(\bt_{n - 1}, \by_{n - 1}) \leq \ve \bigl( \| \bt_{n - 1} \| - \| \by_{n - 1} \| \bigr)
	\end{equation}
	for any $\ve \in[0,1]$. Using it with $\ve_1$, the bound \eqref{Sn-bound} with $\ve_2$ and 
	\begin{equation}
	\sum_{ \substack{ i, j = 1 \\ i < j} }^{n - 1} \bigl| t_i - t_j \bigr| \leq (n - 1) \| \bt_{n-1} \|
	\end{equation}
	we have
	\begin{equation}
	|G| \leq C_2 \exp \pi \biggl( \bigl[ \ve_1 \nu_{g^*} - \nu_g \bigr] \| \bt_{n - 1} \| + \bigl[ \ve_2 \nu_g - \ve_1 \nu_{g^*} \bigr] \| \by_{n - 1} \| - \frac{\ve_2 \nu_g}{c_{n - 1}} \sum_{k = 1}^{n - 2} \| \by_k \| \biggr).
	\end{equation}
	For small enough $\ve_1, \ve_2$, such that
	\begin{equation}
	\ve_1 < \frac{\nu_{g}}{\nu_{g^*}}, \qquad \ve_2 < \ve_1 \frac{\nu_{g^*}}{\nu_g},
	\end{equation} 
	the function from the right is integrable.
\end{proof}

\begin{proof}[Proof of Theorem \ref{thm:QPsi}]
	In the notation of Proposition \ref{prop1} the theorem states that
	\begin{equation}\label{I-Psi}
	I_{\l, \bl_n} = \prod_{j = 1}^n \hat{K}^*(\l - \l_j) \, \Psi_{\bl_n}.
	\end{equation}
	By Proposition \ref{prop1} the function from the left is analytic on compact subsets of domain \eqref{I-dom}, the same is true for the right-hand side, see \cite[Proposition 1]{BDKK2}. Hence, first we prove the statement \eqref{I-Psi} assuming
	\begin{equation}
	\Im \l = \Im \l_j, \qquad j = 1, \dots, n
	\end{equation}
	and then analytically continue it. 
	
	The case $n = 1$
	\begin{equation}
	I_{\l, \l_1} = d_1(g^*) \int_\R dy_1 \, K^*(x_1 - y_1) \, e^{2\pi \imath \l (x_1 - y_1) } \, e^{2\pi \imath \l_1 y_1} = \hat{K}^*(\l - \l_1) \, e^{2\pi\imath \l_1 x_1}
	\end{equation}
	is equivalent to the Fourier transform \eqref{K-fourier}. Then proceed by induction. Due to the absolute convergence of the integral $I_{\l, \bl_n}$ we can interchange the order of integrals in it and use exchange relation \eqref{Q*L}
	\begin{equation}
	\begin{aligned}
	I_{\l, \bl_n} & = Q_{n}^*(\l) \, \Lambda_n(\l_n) \, \Lambda_{n - 1}(\l_{n - 1}) \cdots \Lambda_2(\l_2) \, e^{2\pi \imath \l_1 x_1}  \\[8pt]
	& = \hat{K}^*(\l - \l_n) \, \Lambda_n(\l_n) \, Q_{n - 1}^*(\l) \, \Lambda_{n - 1}(\l_{n - 1}) \cdots \Lambda_2(\l_2) \, e^{2\pi \imath \l_1 x_1} \\[6pt]
	& = \hat{K}^*(\l - \l_n) \, J_{\l, \bl_n}.
	\end{aligned}
	\end{equation}
	By Proposition \ref{prop1} the integral $J_{\l, \bl_n}$ is also absolutely convergent and therefore we make the integrals associated with $\Lambda_n(\l_n)$ in it to be the last ones. The rest integrals give the integral $I_{\l, \bl_{n - 1}}$, for which we use induction assumption and arrive at the statement~\eqref{I-Psi}.
\end{proof}

\subsection{Wave function $g \rightarrow g^*$ symmetry}\label{sec:eigen-g}
\Psisym*

The proof of the symmetry relies on the recursive construction of the wave function
\begin{equation}\label{Psi-L}
\Psi_{\bl_n}(\bx_n) = \Lambda_n(\lambda_n) \, \Psi_{\bl_{n - 1}}(\bx_{n - 1})
\end{equation}
and the exchange relation given by Theorem 4
\begin{equation}\label{L*L}
\Lambda^*_{n}(\l) \, \Lambda_{n - 1}(\rho) = K_{2\hat{g}}(\l - \rho | \hat{\bo} ) \, \Lambda_n(\rho)  \, \Lambda^*_{n - 1}(\l).
\end{equation}
To justify the interchange of integrals appearing in the proof we also state the following proposition.

\begin{restatable}{prop}{prop2}\label{prop2}
	The multiple integrals
	\begin{align}
	\tilde{I}_{\l, \bl_n} &= \Lambda^*_{n + 1}(\l) \, \Lambda_n(\l_n)  \, \Lambda_{n - 1}(\l_{n - 1}) \, \cdots \, \Lambda_2(\l_2) \, e^{2\pi \imath \l_1 x_1}, \\[6pt]
	\tilde{J}_{\l, \bl_n} &= \Lambda_{n + 1}(\l_n) \, \Lambda^*_{n}(\l) \, \Lambda_{n - 1}(\l_{n - 1}) \, \cdots \, \Lambda_2(\l_2) \, e^{2\pi \imath \l_1 x_1}
	\end{align}
	are absolutely convergent under restriction
	\begin{equation}\label{l-rest}
	\Im \l = \Im \l_k, \qquad k = 1, \dots, n.
	\end{equation}
\end{restatable}
\begin{proof}[Proof of Proposition \ref{prop2}]
The integrand of $\tilde{I}_{\l, \bl_n}$ almost coincides with the integrand of the integral $I_{\l, \bl_n}$ \eqref{Iint} from Proposition \ref{prop1} up to additional functions 
\begin{equation}
\prod_{j = 1}^{n + 1} K^*(x_{n + 1} - y_j)
\end{equation}
that only improve convergence since $K$-functions are exponentially bounded \eqref{Kmu-bound}. Hence, $\tilde{I}_{\l, \bl_n}$ is absolutely convergent.

Next consider $\tilde{J}_{\l, \bl_n}$. Denote the groups of integration variables by
\begin{equation}
\by_k = \bigl( y_1^{(k)}, \dots, y_k^{(k)} \bigr).
\end{equation} 
The integral in its full form
\begin{equation}
\begin{aligned}
&\tilde{J}_{\l, \bl_n} = C_{\tilde{J}} \int \prod_{j = 1}^n d\by_j \; \Delta(\by_n) \, K(\bx_{n + 1}, \by_n) \, e^{2\pi \imath \bigl[ \l_n \bbx_{n + 1} + (\l - \l_n) \bby_n \bigr]} \times\\
& \Delta(\by_{n - 1}) \, K^*(\by_n, \by_{n - 1}) \, e^{2\pi\imath (\l_{n - 1} - \l)\bby_{n - 1}} \prod_{j = 1}^{n - 2} \mu(\by_j) \, K(\by_{j + 1}, \by_j) \, e^{2\pi\imath (\l_j - \l_{j + 1}) \bby_j},
\end{aligned}
\end{equation}
where in the last product we assume $\mu(\by_1) \equiv 1$. The constant $C_{\tilde{J}}$ contains all constants $d_k$ from operators and all integrals are over $\R$. Denote the integrand by $F$. Then under restriction \eqref{l-rest} with the help of the bounds \eqref{hvand-b}, \eqref{Kmu-bound} and triangle inequalities we arrive at
\begin{equation}
\begin{aligned}
|F| \leq C  \exp \pi \biggl( &- (n + 1) \nu_g \| \by_n \| + \frac{\nu_g}{2} \sum_{ \substack{ i,j = 1 \\ i \not= j } }^{n} \bigl| y_i^{(n)} - y_j^{(n)} \bigr| + \nu_{g^*} L_n(\by_{n - 1}, \by_n) \\[4pt]
&- \frac{\nu_g }{2} \sum_{ \substack{ i,j = 1 \\ i \not= j } }^{n-1} \bigl| y_i^{(n-1)} - y_j^{(n-1)} \bigr|   + \nu_g S_{n - 1}(\by_1, \dots, \by_{n - 1})\biggr),
\end{aligned}
\end{equation}
where $L_n$ is defined as
\begin{equation}
L_n(\by_{n - 1}, \by_n) = \sum_{ \substack{ i, j = 1 \\ i < j} }^n \bigl| y^{(n)}_i - y^{(n)}_j \bigr| + \sum_{ \substack{ i, j = 1 \\ i < j} }^{n - 1} \bigl| y^{(n-1)}_ i - y^{(n-1)}_j \bigr| - \sum_{i = 1}^{n}\sum_{j = 1}^{n - 1} \bigl| y^{(n)}_i - y^{(n-1)}_j \bigr|
\end{equation}
and $S_{n - 1}$ is defined in \eqref{Srec}. Using the bound \eqref{Ln-bound} with $\ve_1$, the bound \eqref{Sn-bound} with $\ve_2$ and the estimate \eqref{sum<} we have
\begin{equation}
|F| \leq C \exp \pi \biggl( \bigl[ (n - 1) \ve_1 \nu_{g^*} -2 \nu_g \bigr] \| \by_n \| + \bigl[ \ve_2 \nu_g - \ve_1 \nu_{g^*} \bigr] \| \by_{n - 1} \| - \frac{\ve_2}{c_{n - 1}} \sum_{k = 1}^{n - 2} \| \by_k \| \biggr).
\end{equation}
The function from the right is integrable for small enough $\ve_1, \ve_2$ such that
\begin{equation}
(n - 1) \ve_1 \nu_{g^*} < 2 \nu_g, \qquad \ve_2 \nu_g < \ve_1 \nu_{g^*}.
\end{equation}
Therefore, the integral $\tilde{J}_{\l, \bl_n}$ is absolutely convergent.
\end{proof}

\begin{proof}[Proof of Theorem \ref{thm:psi-g}]
The wave function is analytic with respect to $\l_j$, see \cite[Proposition 1]{BDKK2} and remark after it. So, it is sufficient to prove the statement of the theorem for real $\l_j$.

The proof goes by induction. The case $n = 1$ is trivial, since the wave function 
\begin{equation}
\Psi_{\l_1}(x_1) = e^{2\pi \imath \l_1 x_1}
\end{equation}
doesn't depend on $g$. Then assume we proved the $(n - 1)$-particle case 
\begin{equation}\label{indassump}
\Psi_{ \bm{\lambda}_{n - 1} }(\bx_{n - 1}; g ) = \hat{\eta}^{-1}(\bl_{n - 1}) \, \eta^{-1}(\bx_{n - 1}) \, \Psi_{ \bm{\lambda}_{n - 1} }(\bx_{n - 1}; g^* )
\end{equation}
and let us prove the $n$-particle case. Here and in what follows we omit the dependence on periods $\bo$. 

First, note that in terms of the $\Lambda^*$-operator \eqref{L*def} the recursive formula for the function $\Psi_{ \bm{\lambda}_{n} }(\bx_{n}; g^* )$ looks as
\begin{equation}
\Psi_{ \bm{\lambda}_{n} }(\bx_{n}; g^* ) = \eta(\bx_n) \, \Lambda_n^*(\l_n) \, \eta^{-1}(\bx_{n - 1}) \, \Psi_{ \bm{\lambda}_{n - 1} }(\bx_{n - 1}; g^* ).
\end{equation}
Then using induction assumption \eqref{indassump} we have
\begin{equation}
\begin{aligned}
\eta^{-1}(\bx_n) \, \Psi_{ \bm{\lambda}_{n} }(\bx_{n}; g^* ) &= \hat{\eta}(\bl_{n - 1})  \, \Lambda_n^*(\l_n) \, \Psi_{ \bm{\lambda}_{n - 1} }(\bx_{n - 1}; g ) \\[6pt]
& = \hat{\eta}(\bl_{n - 1})  \, \Lambda_n^*(\l_n) \, \Lambda_{n - 1}(\l_{n - 1}) \, \cdots \, \Lambda_2(\l_2) \, e^{2\pi \imath \l_1 x_1}.
\end{aligned}
\end{equation}
The multiple integral in the last line is absolutely convergent due to Proposition \ref{prop2}. Therefore, we can change order of integrals in it and use exchange relation \eqref{L*L} to obtain
\begin{equation}
\begin{aligned}
\eta^{-1}(\bx_n) \, \Psi_{ \bm{\lambda}_{n} }(\bx_{n}; g^* ) &= \hat{\eta}(\bl_{n - 1})  \, K_{2\hat{g}}(\l_n - \l_{n - 1} | \hat{\bo})  \\[6pt]
&\times \Lambda_n(\l_{n - 1}) \, \Lambda^*_{n - 1}(\l_{n})  \, \Lambda_{n - 2}(\l_{n - 2})\, \cdots \, \Lambda_2(\l_2) \, e^{2\pi \imath \l_1 x_1}.
\end{aligned}
\end{equation}
Again by Proposition \ref{prop2} the integral from the right is absolutely convergent. Then we proceed by exchanging $\Lambda^*$-operator with all $\Lambda$-operators from the right and arrive at
\begin{equation}
\begin{aligned}
\eta^{-1}(\bx_n) \, \Psi_{ \bm{\lambda}_{n} }(\bx_{n}; g^* ) &= \hat{\eta}(\bl_{n - 1})  \, \prod_{j = 1}^{{n - 1}} K_{2\hat{g}}(\l_n - \l_{j} | \hat{\bo}) \, \Psi_{ \bm{\lambda}_{n} }(\bx_{n}; g).
\end{aligned}
\end{equation}
Since
\begin{align}
K_{2\hat{g}}(\l_n - \l_j | \hat{\bo}) = S^{-1}_2(\imath \l_n - \imath \l_j + \hat{g}^*|\hat{\bo}) \, S^{-1}_2(\imath \l_j - \imath \l_n + \hat{g}^*|\hat{\bo})
\end{align}
due to the definition \eqref{coef} we have
\begin{equation}
\hat{\eta}(\bl_{n - 1})  \, \prod_{j = 1}^{{n - 1}} K_{2\hat{g}}(\l_n - \l_{j} | \hat{\bo}) = \hat{\eta}(\bl_n).
\end{equation}
Thus, we proved the statement for $n$-particle case.
\end{proof}

Now we prove one important corollary of the relations \eqref{psi-g}, \eqref{duality} and of the main result of the paper \cite{HR3} concerning the asymptotics of the function \eqref{Edef} 
\begin{align}
E_{\bl_n}(\bx_n) := E_{\bl_n}(\bx_n; g | \bo) =  e^{-\frac{\imath \hat{g}g^*}{4}n(n - 1)} \, \mu'(\bx_n) \, \hat{\mu}'(\bl_n) \, \Psi_{\bl_n}(\bx_n).
\end{align}

\propAsymp*

\begin{proof}
	Using space-spectral duality \eqref{duality} from \eqref{Easymp1} we obtain the asymptotics \eqref{Easymp} under the assumption
	\begin{equation}
	\Re \hat{g}^* \in \bigl(0, \max(\hat{\o}_1, \hat{\o}_2)\bigr],
	\end{equation} 
	or equivalently
	\begin{equation}\label{g-cond}
	\Re g \in \bigl[ \min(\o_1, \o_2), \o_1 + \o_2 \bigr).
	\end{equation}
	Due to the relation \eqref{psi-g} and reflection formula \eqref{A3} we have the symmetries
	\begin{equation}
	E_{\bl_n}(\bx_n; g) = E_{\bl_n}(\bx_n; g^*), \qquad E^{\mathrm{as}}_{\bl_n}(\bx_n; g) = E^{\mathrm{as}}_{\bl_n}(\bx_n; g^*).
	\end{equation}
	In the case $\Re g \in (0,  \min(\o_1, \o_2))$, that is outside of the domain \eqref{g-cond},
	\begin{equation}
	\Re g^* \in \bigl( \max(\o_1, \o_2), \o_1 + \o_2 \bigr) \subset \bigl[ \min(\o_1, \o_2), \o_1 + \o_2 \bigr).
	\end{equation}
	Hence, the function $E_{\bl_n}(\bx_n; g^*)$ has asymptotics $E^{\mathrm{as}}_{\bl_n}(\bx_n; g^*)$.
\end{proof}

\setcounter{equation}{0}
\section{ Noumi-Sano difference operators}\label{sec:NS}
\propNS*

\begin{proof}	
The proof consists of an explicit calculation of residues in the right-hand side of \rf{N8} which makes sense due to the conditions \rf{N11}. Introduce one more notation for the hyperbolic Pochhammer symbol
\begin{equation} \label{N12} 
	[x]_{m,k}=\frac{S_2(x)}{S_2(x+m\o_1+k\o_2)},\qquad m,k\in\Z. 
\end{equation}
Due to factorization formula for the double sine function \eqref{Sfact}
\begin{equation}\label{N13} 
	[x]_{m,k}=(-1)^{mk} \, [x|\o_1]_m \; [x|\o_2]_k
\end{equation}
where Pochhammer symbols in the right-hand side of \rf{N13} are defined in \rf{N0}.

The kernel of the $Q^*$-operator \eqref{Q*ker} explicitly reads
\begin{multline} \label{N14}
	Q^* \Bigl(\bx_n+\frac{\imath g}{2}\bbe_n,
	\by_n;\l \Bigr) = e^{2\pi\imath \l (\bbx_n-\bby_n)-\pi \l n g } \\[6pt]
	\times \prod_{\substack{i,j=1 \\ i\not=j}}^n S_2(\imath x_{ij}+g) \, S_2(\imath y_{ij})\prod_{i,j=1}^n S_2^{-1}(\imath x_i- \imath y_j) \, S_2^{-1} (\imath y_i- \imath x_j+g)
\end{multline} 
where for brevity we denoted $x_{ij} = x_i - x_j$.  Because of the reflection formula \eqref{trig4} the product of functions $S_2(\imath y_{ij})$ doesn't have poles.
Then due to the formula \eqref{trig5} the residue 
\begin{equation}\label{N15}
	\Res_{\substack{y_1=x_1-\imath(m_1\o_1+k_1\o_2)\\[2pt] \ldots \\[2pt] y_n=x_n-\imath(m_n\o_1+k_n\o_2)}} \, Q^*\Bigl(\bx_n+\frac{\imath g}{2}\bbe_n, \by_n;\l \Bigr)f(\by_n) 
\end{equation}
equals to
\begin{multline} \label{N16} 
	\biggl(\frac{ \imath \sqrt{\o_1\o_2}}{2\pi}\biggr)^n e^{-\pi\l(2|\bm{m}|\o_1+2|\bk|\o_2+ng)}\prod_{i=1}^n \frac{ S_2^{-1}(g+m_i\o_1+k_i\o_2)}{[\o_1+\o_2]_{m_i,k_i}} \\[6pt]
	\times \prod_{\substack{i,j = 1 \\ i\not=j}}^n \frac{S_2(\imath x_{ij}+m_{ij}\o_1+k_{ij}\o_2)}{S_2(\imath x_{ij}-m_j\o_1-k_j\o_2)}\frac{S_2(\imath x_{ij}+g)}{S_2(\imath x_{ij}+g+m_i\o_1+k_i\o_2)}\\[6pt]
	\times f \bigl(x_1-\imath(m_1\o_1+k_1\o_2),\ldots,x_n-\imath(m_n\o_1+k_n\o_2) \bigr).
\end{multline}
The first two lines of \rf{N16} can be rewritten in terms of Pochhammer symbols \eqref{N12} as
\begin{equation}\label{N24}
\begin{aligned}
&\biggl(\frac{\imath \sqrt{\o_1\o_2}}{2\pi S_2(g)}\biggr)^n e^{-\pi\l(2|\bm{m}|\o_1+2|\bk|\o_2+ng)}\\[6pt]
&\times  \prod_{i=1}^n \frac{[g+m_i\o_1+k_i\o_2]_{m_i,k_i}}{[\o_1+\o_2]_{m_i,k_i}} \prod_{\substack{i,j = 1 \\ i\not=j}}^n
\frac{[\imath x_{ij}+g]_{m_i,k_i}}{[\imath x_{ij}-m_j\o_1-k_j\o_2]_{m_i,k_i}}.
\end{aligned}
\end{equation} 	 
 Due to the reflection formula \eqref{A3} the following relation holds
\begin{equation}
[\o_1 + \o_2]_{m_i, k_i} = [- m_i \o_1 - k_i \o_2]_{m_i, k_i}
\end{equation} 
and therefore two products in \eqref{N24} can be unified. 
Finally, using factorization \eqref{N13} we arrive at 
\begin{multline}\label{N17}
		\Res_{\substack{y_1=x_1-\imath(m_1\o_1+k_1\o_2)\\[2pt] \ldots \\[2pt] y_n=x_n-\imath(m_n\o_1+k_n\o_2)}} \, Q^*\Bigl(\bx_n+\frac{\imath g}{2}\bbe_n, \by_n;\l \Bigr)f(\by_n)  \\[8pt]
	= \biggl( \frac{\imath \sqrt{\o_1\o_2}}{2\pi S_2(g)}\biggr)^n  e^{-\pi\l(2|\bm{m}|\o_1+2|\bk|\o_2+ng)} \prod_{i,j=1}^n\frac{[\imath x_{ij}+g|\o_1]_{m_i}}{[\imath x_{ij}-m_j\o_1|\o_1]_{m_i}}
	\prod_{i,j=1}^n\frac{[\imath x_{ij}+g|\o_2]_{k_i}}{[\imath x_{ij}-k_j\o_2|\o_2]_{k_i}}\\[6pt]
	\times f \bigl(x_1-\imath(m_1\o_1+k_1\o_2),\ldots,x_n-\imath(m_n\o_1+k_n\o_2) \bigr).
\end{multline}
Summing up over all $\bm{m}$ and $\bm{k}$ with $|\bm{m}|=p$ and $|\bm{k}|=q$ we obtain the statement of the proposition. 
\end{proof}

\propQNS*
  Here 
\begin{equation*}  c^{(2)}(\bx_n;\l)=N^{(2)}(\bx_n;\l) \, \mathbf{1} \end{equation*} \cbk
\begin{proof}
The arguments are similar to the ones given in the previous proof, but being used in opposite direction.
 Calculating the integral 
 \begin{equation}\label{N18}  \int_{\R^n} d\by_n \, Q^*\Bigl(\bx_n+\frac{\imath g}{2}\bbe_n,
 \by_n;\l \Bigr)f(\by_n) \end{equation}
by residue technique, which is supposed to be applicable, we meet simple poles of two kinds.
The poles of the first kind are
 \begin{equation}\label{N19} \imath y_1= \imath x_{\sigma_1}+\frac{g}{2}+m_1\o_1+k_1\o_2, \qquad \ldots \qquad \imath y_n=\imath x_{\sigma_n}+\frac{g}{2}+m_n\o_1+k_n\o_2,\end{equation}
for some permutation $\sigma\in S_n$. Since the function $f(\bx_n)$ is assumed to be symmetric their contribution doesn't depend on a permutation and are computed above. The sum over permutations gives additional factor $n!$ in the final answer.
  
In the poles of the second kind the indeces of variables $x_j$ may coincide. We claim that for $\imath \o_2$-periodic function $f(\bx_n)$ their sum vanishes. Assume for definiteness that the pole is at the point
\begin{equation}\label{N20} \imath y_1=\imath x_{1}+\frac{g}{2}+m_1\o_1+k_1\o_2, \qquad \imath y_2=\imath x_{1}+\frac{g}{2}+m_2\o_1+k_2\o_2, \qquad \ldots \end{equation}
  Consider in addition the point
  \begin{equation}\label{N211} \imath y_1=\imath x_{1}+\frac{g}{2}+m_1\o_1+k_2\o_2, \qquad \imath y_2=\imath x_{1}+\frac{g}{2}+m_2\o_1+k_1\o_2, \qquad \ldots \end{equation}
  where the coefficients $k_1$ and $k_2$ are exchanged. We claim that sum of the residues at these two points vanish. This can be seen from a factorized form of the coefficients in residues. The minus sign comes from the measure function  $\Delta(\by_n)$, which at the point \eqref{N20} contains 
  \begin{equation}
  S(\imath y_{12}) \, S(\imath y_{21}) = 4 (-1)^{m_1 + m_2 + k_1 + k_2 + 1} \sin \frac{\pi (m_1 - m_2) \o_1}{\o_2}  \, \sin \frac{\pi (k_{1} - k_2) \o_2}{\o_1}.
  \end{equation}

  Thus, the residue calculation gives 
  \begin{multline}\label{N212}
  \int_{\R^n} d\by_n \, Q^*\Bigl(\bx_n+\frac{\imath g}{2}\bbe_n,
  	\by_n;\l \Bigr)f(\by_n) = n!\biggl(\frac{\sqrt{\o_1\o_2}}{ S_2(g)}\biggr)^ne^{-\pi\l ng}\\[6pt]
  \sum_{p,q\geq 0}e^{-2\pi\l(p\o_1+q\o_2)}	\sum_{\substack{|\bm{m}|=p,\\[2pt] |\bk|=q}} \,  \prod_{i,j=1}^n\frac{[\imath x_{ij}+g|\o_1]_{m_i}}{[\imath x_{ij}-m_j\o_1|\o_1]_{m_i}}
  \prod_{i,j=1}^n\frac{[\imath x_{ij}+g|\o_2]_{k_i}}{[\imath x_{ij}-k_j\o_2|\o_2]_{k_i}}\\[6pt]
  \times f\bigl(x_1-\imath(m_1\o_1+k_1\o_2),\ldots,x_n-\imath(m_n\o_1+k_n\o_2)\bigr).
  \end{multline}
This relation with the use of \eqref{dconst}, \rf{N3} and \rf{N5} can be rewritten in a factorized form as 
 \begin{equation}\label{N21} 
 \bigl(Q^*_n(\l) f \bigr) \Bigl(\bx_n+\frac{\imath g}{2}\bbe_n\Bigr)=e^{-\pi\l ng}N^{(1)}(\bx_n ;\l)N^{(2)}(\bx_n;\l) f(\bx_n).
 \end{equation}
Since $f(\bx_n)$ is supposed to be $\imath\o_2$-periodic, the relation \rf{N21} can be also rewritten as
\begin{equation}\label{N22} \bigl(Q^*_n(\l) f \bigr) \Bigl(\bx_n+\frac{\imath g}{2}\bbe_n\Bigr)=e^{-\pi\l ng}N^{(1)}(\bx_n ;\l) \, c^{(2)}(\bx_n;\l) f(\bx_n)
\end{equation}
where 
\begin{equation} c^{(2)}(\bx_n;\l)=N^{(2)}(\bx_n;\l) \, {\mathbf {1}}=\sum_{q\geq 0}(-1)^q e^{-2\pi\l q\o_2}\sum_{|\bk|=q}\prod_{i,j=1}^n\frac{[\imath x_{ij}+g|\o_2]_{k_i}}{[\imath x_{ij}-k_j\o_2|\o_2]_{k_i}}.
\end{equation}

\end{proof}
 
\section*{Acknowledgments}
The authors thank N. Nekrasov, N. Reshetikhin and  V. Spiridonov for stimulating discussions and interest in the work.

The work of N. Belousov and S. Derkachov was supported by Russian Science Foundation, project No. 23-11-00311, used for the proof of statements of Section 2 and Appendices A, B, C. The work of S. Kharchev was supported by Russian Science Foundation, project No. 20-12-00195, used for the
proof of statements of Section 3. The work of S. Khoroshkin (Section 4) was supported by the International Laboratory of Cluster Geometry of National Research University Higher School of Economics, Russian Federation Government grant, ag. No. 075-15-2021-608 dated 08.06.2021. 

\setcounter{equation}{0}

\section*{Appendix}
\appendix
\section{The double sine function} \label{AppA}
The  double sine  function $S_2(z):=S_2(z|\bo)$, see \cite{Ku} and references therein, is a meromorphic function that satisfies two functional relations
\begin{equation}\label{trig3}  \frac{S_2(z)}{S_2(z+\o_1)}=2\sin \frac{\pi z}{\o_2},\qquad \frac{S_2(z)}{S_2(z+\o_2)}=2\sin \frac{\pi z}{\o_1}
\end{equation}
and inversion relation
\begin{equation} \label{trig4} S_2(z)S_2(-z)=-4\sin\frac{\pi z}{\o_1}\sin\frac{\pi z}{\o_2},\end{equation}
or equivalently
\begin{equation}\label{A3} S_2(z)S_2(\o_1+\o_2-z)=1. \end{equation}
The factorization formula
\begin{equation} \label{Sfact}
S_2(z + m \o_1 + k \o_2) = (-1)^{mk} \, \frac{S_2(z + m \o_1) \, S_2(z + k \o_2)}{S_2(z)}
\end{equation}
follows from \eqref{trig3}.
The function $S_2(z)$ is a meromorphic function of $z$ with poles at
\begin{equation}\label{A1a} z_{m,k} = m\o_1 + k\o_2, \qquad m,k \geq 1\end{equation}
and zeros at
\begin{equation}\label{A1b} z_{-m,-k}=-m\o_1-k\o_2,\qquad m,k\geq 0. \end{equation}
For $\o_1/\o_2 \not\in \mathbb{Q}$ all poles and zeros are simple. The residues of $S_2(z)$ and $S^{-1}_2(z)$ at these points are
\begin{align}
\underset{z = z_{m,k}}{\Res} \, S_2(z) = \frac{\sqrt{\o_1\o_2}}{2\pi}\frac{(-1)^{mk}}{\prod\limits_{s=1}^{m - 1}2\sin\dfrac{\pi s\o_1}{\o_2}\prod\limits_{l=1}^{k - 1}2\sin\dfrac{\pi l\o_2}{\o_1}},
\\[10pt]
\label{trig5} \underset{z = z_{-m,-k}}{\Res} \, S^{-1}_2(z) = \frac{\sqrt{\o_1\o_2}}{2\pi}\frac{(-1)^{mk+m+k}}{\prod\limits_{s=1}^m2\sin\dfrac{\pi s\o_1}{\o_2}\prod\limits_{l=1}^k2\sin\dfrac{\pi l\o_2}{\o_1}}.
\end{align}
In the analytic region $ \Re z \in ( 0, \Re(\omega_1 + \omega_2) )$ we have the following integral representation for the logarithm of $S_2(z)$
\begin{equation}\label{S2-int}
\ln S_2 (z) = \int_0^\infty \frac{dt}{2t} \left( \frac{\sh \left[ (2z - \omega_1 - \omega_2)t \right]}{ \sh (\omega_1 t) \sh (\omega_2 t) } - \frac{ 2z - \omega_1 - \omega_2 }{ \omega_1 \omega_ 2 t } \right).
\end{equation}
It is clear from this representation that the double sine function is homogeneous
\begin{equation}\label{S-hom}
S_2( \gamma z | \gamma\o_1, \gamma \o_2 ) = S_2(z|\o_1, \o_2), \qquad \gamma \in (0, \infty)
\end{equation}
and invariant under permutation of periods
\begin{equation}\label{A6}
S_2(z| \o_1, \o_2) = S_2(z | \o_2, \o_1).
\end{equation}
The double sine function can be expressed through the Barnes double Gamma function $\Gamma_2(z|\bo)$ \cite{B},
\begin{equation}
S_2(z|\bo)=\Gamma_2(\o_1+\o_2-z|\bo)\Gamma_2^{-1}(z|\bo),
\end{equation}
and its properties follow from the corresponding properties of the double Gamma function.
It is also connected to the Ruijsenaars hyperbolic Gamma function $G(z|\bo)$ \cite{R2}
\begin{equation} \label{S-G}
G(z|\bo) = S_2\Bigl(\imath z + \frac{\o_1 + \o_2}{2} \,\Big|\, \bo \Bigr)
\end{equation}
and to the Faddeev quantum dilogarithm $\gamma(z|\bo)$ \cite{F}
\begin{equation} 
\gamma(z|\bo) = S_2\Bigl(-\imath z + \frac{\o_1+\o_2}{2}\, \Big|\, \bo\Bigr) \exp \Bigl( \frac{\imath \pi}{2\o_1 \o_2} \Bigl[z^2 + \frac{\o_1^2+\o_2^2}{12} \Bigr]\Bigr).
\end{equation}
Both $G(z|\bo)$ and $\gamma(z|\bo)$ were investigated independently.

In the paper we deal only with ratios of double sine functions denoted by $\mu(x)$ \eqref{I5} and $K(x)$ \eqref{I6}
\begin{equation}\label{B1}
\begin{split}\mu(x)& =S_2(\imath x)S_2^{-1} (\imath x+g),\\[6pt]
	K(x)& =  S_2\left(\imath x+\frac{\o_1+\o_2}{2}+\frac{g}{2}\right)S_2^{-1}\left(\imath x+\frac{\o_1+\o_2}{2}-\frac{g}{2}\right).
\end{split}
\end{equation}
Now we will give the key asymptotic formulas and bounds for them, which were derived in \cite[Appendices A, B]{BDKK} from the known results for the double Gamma function. In what follows we assume conditions \eqref{I0a}, \eqref{I0b}
\begin{equation}
\Re \o_j > 0, \qquad 0 < \Re g < \Re\o_1 + \Re \o_2, \qquad \nu_g = \Re \hat{g} > 0,
\end{equation}
where we denoted
\begin{equation}
\hat{g} = \frac{g}{\o_1 \o_2}.
\end{equation}
The functions $\mu(x)$ and $K(x)$ \eqref{B1} with $x \in \R$ have the following asymptotics
\begin{equation}\label{Kmu-asymp}
\mu(x) \sim e^{\pi \hat{g} | x | \pm \imath \frac{\pi \hat{g} g^*}{2} }, \qquad K(x) \sim e^{- \pi \hat{g} |x|}, \qquad x\rightarrow \pm \infty.
\end{equation}
and bounds
\begin{equation} \label{Kmu-bound}
|\mu(x)| \leq C e^{\pi\nu_g |x|}, \qquad |K(x)| \leq C e^{-\pi\nu_g |x|},  \qquad x \in \R
\end{equation}
where $C$ is a positive constant uniform for compact subsets of parameters $\bo, g$ preserving the mentioned conditions, see \cite[eq.(B.3)]{BDKK}.

Another key result that we need in the paper is the following Fourier transform formula given in \cite[Proposition C.1]{R3}, which we rewrite in terms of the double sine function using connection formula \eqref{S-G}.  This Fourier transform can be already found in \cite{FKV,PT}.
\begin{proposition*}\cite{R3}
For real positive periods $\o_1, \o_2$ we have
\begin{equation}
\begin{aligned}
&\int_{\mathbb{R}} dx \, e^{\frac{2\pi\imath}{\o_1 \o_2} y x} S_2\Bigl(\imath x - \imath \nu + \frac{\o_1 + \o_2}{2} \Bigr) S_2^{-1} \Bigl( \imath x - \imath \rho + \frac{\o_1 + \o_2}{2} \Bigr) \\[6pt]
&= \sqrt{\o_1 \o_2} \, e^{ \frac{\pi\imath}{\o_1 \o_2} y (\nu + \rho) } S_2(\imath \rho - \imath \nu) \, S_2^{-1}\Bigl(\imath y + \frac{\imath(\rho - \nu)}{2} \Bigr) \, S_2^{-1}\Bigl(-\imath y + \frac{\imath(\rho - \nu)}{2} \Bigr),
\end{aligned}
\end{equation}
while the parameters $\nu, \rho, y$ satisfy the conditions
\begin{equation}\label{A19}
-\frac{\o_1 + \o_2}{2} < \Im \rho < \Im \nu < \frac{\o_1 + \o_2}{2}, \qquad |\Im y | < \Im \frac{\nu - \rho}{2}.
\end{equation}
\end{proposition*}
In the special case
\begin{equation}
\nu = \frac{\imath g}{2}, \qquad \rho = -\frac{\imath g}{2}
\end{equation}
taking $y = \o_1 \o_2 \l$ and using homogeneity of the double sine \eqref{S-hom} (with $\gamma = \o_1 \o_2$) we arrive at the Fourier transform formula for the function $K(x)$ \eqref{B1}
\begin{equation}\label{K-fourier}
\int_{\mathbb{R}} dx \; e^{2 \pi \imath  \lambda x }  \K (x) = \sqrt{\omega_1 \omega_2} \, S_2(g) \, \KK (\lambda),
\end{equation}
where $| \Im \l | < \Re \hat{g}/2$ and conditions \eqref{A19} are satisfied due to the inequalities on the coupling constant $g$ \eqref{I0a}, \eqref{I0b}. Here we recall the notations
\begin{equation}
\hat{K}(\l) = K_{\hat{g}^*}(\l|\hat{\bo}), \qquad \hat{g}^* = \frac{g^*}{\o_1\o_2}, \qquad \hat{\bo} = \Bigl( \frac{1}{\o_2}, \frac{1}{\o_1} \Bigr).
\end{equation}
Note that the right hand side of \eqref{K-fourier} is analytic function of $\o_1, \o_2$ in the domain $\Re \o_j > 0$. The integral from the left is also analytic with respect to periods. Indeed, due to the bound \eqref{Kmu-bound} it is absolutely convergent uniformly on compact sets of parameters $\bo, g$ preserving the conditions \eqref{I0a}, \eqref{I0b}. Hence, the formula \eqref{K-fourier} also holds for complex periods under the mentioned conditions.

\setcounter{equation}{0}
\section{A degeneration of Rains integral identity}\label{AppB}

\subsection{Hyperbolic  $A_n\rightleftarrows A_m$ identity}

Keep the assumptions $\Re \o_1>0$, $\Re \o_2>0$ and denote
\begin{equation}q=\frac{\o_1+\o_2}{2}, \qquad  \eta=  \Re \frac{2\pi q}{\o_1\o_2} > 0.
\end{equation}
In this Appendix it is convenient to use the following notations
\begin{equation}\gamma^{(2)}(z) =S_2^{-1}(z|\omega_1,\omega_2),\qquad \gamma^{(2)}( a+u,b-u)= \gamma^{(2)}( a+u)\gamma^{(2)}(b-u)\end{equation}
and
\begin{equation}
f(\pm z + c) = f(z + c) \, f(-z + c).
\end{equation}
Assume that $a$ and $b$ are in the region of analyticity of the double sine function, and  
\begin{equation}\label{as0a}
\a=\Re\frac{a+b}{\o_1\o_2}>0,\qquad \beta=\Re \frac{2q-a-b}{\o_1\o_2}>0. \end{equation}
The asymptotical bounds and global analytical properties of the double sine function imply the following lemma, see \cite[eq. (A.20), (A.29)]{BDKK} for the details.
\begin{lemma}\label{l1}  For any $u \in \R$ we have the uniform bounds
	\begin{align}\label{as0b}C_1 e^{-\b|u|} <| \gamma^{(2)}( a+\imath u,b-\imath u)|<C_2 e^{-\b|u|},\qquad C_1,C_2>0.\end{align}
\end{lemma}

Here is the hyperbolic limit \cite[Theorem 4.6]{Rains2} of $A_n-A_m$ Rains integral identity
\cite[Theorem 4.1]{Rains1} (see also \cite{SS})
\begin{multline}\label{id1}
	\frac{1}{(n+1)!}\int_{\R^n}
	\frac{\prod_{j=1}^{n+1}\prod_{\ell=1}^{n+m+2}\gamma^{(2)}(g_\ell+\imath u_j\,,f_\ell-\imath u_j) }
	{\prod_{1\leq j<k\leq n+1}\gamma^{(2)}(\pm\imath (u_j- u_k))
	} \prod_{j=1}^n \frac{du_j}{ \sqrt{\omega_1\omega_2}} = \prod_{j,k=1}^{n+m+2}\gamma^{(2)}(g_j+f_k)\\[6pt]
	\times \frac{1}{(m+1)!}\int_{\R^m}
	\frac{\prod_{j=1}^{m+1}\prod_{\ell=1}^{n+m+2}\gamma^{(2)}(g'_\ell+ \imath u_j,f'_\ell-\imath u_j)}
	{\prod_{1\leq j<k\leq m+1}\gamma^{(2)}(\pm\imath (u_j- u_k))
	} \prod_{j=1}^m \frac{du_j}{ \sqrt{\omega_1\omega_2}}
\end{multline}
where integration variables satisfy the relations
\begin{equation}\label{balans0}\sum_{j=1}^{n+1}u_j=0, \qquad\quad  
\sum_{j=1}^{m+1}u_j=0\end{equation}
in the first and the second integrals correspondingly.
External parameters $g_{\ell}$ and $f_{\ell}$ obey the following
balancing condition:
\begin{align}\label{balans}
	G+F=2(m+1)q, \qquad G=\sum_{\ell=1}^{n+m+2}g_\ell,\qquad F=\sum_{\ell=1}^{n+m+2}f_\ell.
\end{align}
Parameters $g'_{\ell}$ and $f'_{\ell}$ are connected with
$g_{\ell}$ and $f_{\ell}$ by means of the following transformation
\begin{equation}\label{balans1}
g'_{\ell} =\frac{G}{ m+1}-g_\ell, \qquad     f'_{\ell} =\frac{F}{ m+1}-f_\ell, \qquad \ell=1\ldots, n+m+2.
\end{equation}

Assume that all the parameters $f_l$, $g_l$, $f'_l$, $g'_l$, have real positive parts and  the sums $f_l+g_l$ and $f'_l+g'_l$ are in the region of analyticity of the double sine function. It is achieved, e.g., once these parameters are in a vicinity of the middle point 
\begin{equation}\label{as0c}f_l=g_l=\frac{m+1}{n+m+2}q.\end{equation} 
Then due to Lemma \ref{l1} and balancing conditions  the integrand of the left-hand side of \rf{id1} can be bounded by the function
\begin{equation} C\exp\eta\biggl(-(n+1)\sum_{j=1}^{n+1}|u_j|+\sum_{\substack{i,j=1 \\ i<j}}^{n + 1}|u_i-u_j|\biggr) \leq 
C'\exp\eta\biggl(-\sum_{j=1}^{n+1}|u_j|\biggr).\end{equation}
Analogous bound we have for the right-hand side of \rf{id1}, so that this identity has a non-empty region of parameters where it is presented by convergent integrals. \cbk
\subsection{Removing the  condition $\textstyle\sum_{j}\,u_j=0$}
To remove conditions \eqref{balans0} we shift the external parameters
\begin{align}
	g_{\ell} \to g_{\ell} +\imath L, \qquad  f_{\ell} \to f_{\ell} -\imath L,\qquad L >0,
\end{align} 
and then calculate the leading asymptotic of both sides as $L \to  \infty$.

Denote the domain 
\begin{equation}
	D_j=\{(u_1,\ldots,u_{n+1})\in\R^{n+1} \colon  {u_j} \geq {u_k}, \, \forall k\not=j\}.\end{equation} 
Due to $S_{n+1}$ symmetry  the integrand in the left-hand side, the integral is $n+1$ times the same integral over the region $D_{n+1}$. Similarly for the right-hand side, so that we replace the identity \rf{id1} by the same equality of integrals over the regions $D_{n+1}$ and $D_{m+1}$ substituting $\frac{n+1}{(n+1)!} = \frac{1}{n!}$
in front of the left-hand side and
$\frac{m+1}{(m+1)!} = \frac{1}{m!}$ in front of the right-hand side.

Now consider the left-hand side. Change the integration variables
\begin{align}\label{as5a}
	u_{j} =  v_{j} - L, \qquad j=1,\ldots, n; \qquad  u_{n+1} = v_{n+1} +n L.
\end{align}
The integrand transforms as follows
\begin{align}\label{as6}
	\frac{\prod_{j=1}^{n+1}\prod_{\ell=1}^{n+m+2}\gamma^{(2)}(g_\ell+ \imath u_j\,,f_\ell-\imath u_j) }
	{\prod_{1\leq j<k\leq n+1}\gamma^{(2)}(\pm\imath (u_j- u_k))} =
	\frac{\prod_{j=1}^{n}\prod_{\ell=1}^{n+m+2}\gamma^{(2)}(g_\ell+ \imath v_j\,,f_\ell-\imath v_j) }
	{\prod_{1\leq j<k\leq n}\gamma^{(2)}(\pm\imath (v_j- v_k))} \\[8pt] \label{as7}
	\times \frac{\prod_{\ell=1}^{n+m+2}\gamma^{(2)}(g_\ell + \imath v_{n+1}  +\imath(n+1)L\,,f_\ell-\imath v_{n+1} -\imath(n+1)L) }
	{\prod_{1\leq j\leq n}\gamma^{(2)}(\pm\imath (v_j - v_{n+1}  -  (n+1)L))}.
\end{align}
Next we recall the asymptotic \cite[eq. (A.19)]{BDKK} 
\begin{align} \label{as1}
	\gamma^{(2)}(z|\omega_1,\omega_2) = e^{\mp\frac{\imath\pi}{2} B_{2,2}(z|\omega_1,\omega_2)} \Bigl(1+O\bigl(z^{-1}\bigr)\Bigr)
\end{align}
for $\pm \Im(z)>0$ and $|z| \to \infty$ along a vertical strip of a fixed width, 
where $B_{2,2}(z|\omega_1,\omega_2)$ is a multiple Bernoulli polynomial
\begin{align}
B_{2,2}(z|\omega_1,\omega_2) = \frac{z^2}{\omega_1\omega_2} -
\frac{\omega_1+\omega_2}{\omega_1\omega_2}\,z +
\frac{\omega^2_1+3\omega_1\omega_2+\omega^2_2}{6\omega_1\omega_2} =
\frac{(z-q)^2}{\omega_1\omega_2}  - \frac{\omega^2_1+\omega^2_2}{12\omega_1\omega_2}
\end{align}
and $2q = \omega_1+\omega_2$. We use the following corollary of \rf{as1}
\begin{equation}\label{as2} \g^{(2)} (z+a)\g^{(2)}(-z+b)=e^{\pm \frac{ \imath \pi}{\o_1\o_2} (2q-a-b)(z+(a-b)/2)} \Bigl(1+O\bigl(z^{-1}\bigr)\Bigr)\end{equation}
for $\pm \Im(z+a)>0$, $\Im(z-b)>0$ and $|z| \to \infty$ along a  vertical strip of a fixed width.

\cbk The following inequalitites are valid in $D_{n+1}$ for real positive $L$
\begin{align}\label{as2a}
	A=v_{n+1} +(n+1)L > L, \qquad  B_j=v_{n+1} - v_{j} +(n+1)L > 0.
\end{align}
Let us take $g_{\ell}\,,f_{\ell} \in \mathbb{R}$.
Then, by \rf{as2} we have the following asymptotic as $L \to \infty$
\begin{equation}
\label{as3}\begin{split}
	&\frac{\prod_{\ell=1}^{n+m+2}\gamma^{(2)}(g_\ell + \imath v_{n+1}  + \imath (n+1)L\,,f_\ell-\imath v_{n+1} -\imath (n+1)L) }
	{\prod_{1\leq j\leq n}\gamma^{(2)}(\pm\imath (v_j - v_{n+1}  - (n+1)L))} \\[10pt]
	&= \frac{e^{-\frac{\imath\pi}{2\omega_1\omega_2}
			\sum_{\ell=1}^{n+m+2}
			\left(g_\ell + f_\ell-2q\right)
			\left(g_\ell - f_\ell+2\imath v_{n+1}+\imath  2(n+1)L\right)} }
	{e^{\frac{\imath \pi}{2\omega_1\omega_2}
			\sum_{j=1}^{n}(2q)
			\left(-2\imath v_j + 2\imath v_{n+1}  + 2(n+1)\imath L\right)}}.
	\biggl(1+O\bigl(A^{-1}\bigr)+\sum_{j=1}^nO\bigl(B_j^{-1}\bigr)\biggr)\end{split}\end{equation}
Using the balancing conditions \rf{balans0} and \rf{balans} we can simplify the exponent in the right-hand side of \rf{as3} and rewrite it as
\begin{equation} \label{as4}
\exp \biggl(-\frac{2q(n+1)L}{\o_1\o_2} \biggr) \, \exp\frac{\imath\pi}{2\omega_1\omega_2}\biggl(2q\left(G-F\right)+\sum_{\ell=1}^{n+m+2}\left(- g^2_\ell + f^2_\ell\right) \biggr).
\end{equation}
The error term $O(1/A)$ can be replaced by $O(1/L)$ due to \rf{as2a}, errow terms $O(1/B_j)$ are not small only in a vicinity of the zero point of the measure function, which can be droped in the total integral or replaced by $o(1)$. Thus, we have finally the estimate for $L\to\infty$
\begin{equation}\label{as5}\begin{split}&\frac{\prod_{\ell=1}^{n+m+2}\gamma^{(2)}(g_\ell + \imath v_{n+1}  + \imath (n+1)L\,,f_\ell-\imath v_{n+1} -\imath (n+1)L) }
	{\prod_{1\leq j\leq n}\gamma^{(2)}(\pm\imath(v_j - v_{n+1} - (n+1)L))} \\[8pt]
	&=\exp\biggl(-\frac{2q(n+1)L}{\o_1\o_2} \biggr) \, \exp\frac{\imath \pi}{2\omega_1\omega_2}\biggl(2q\left(G-F\right)+\sum_{\ell=1}^{n+m+2}\bigl(- g^2_\ell + f^2_\ell\bigr)
	\biggr)\bigl(1+o(1)\bigr).	 \end{split}
\end{equation} 
In the same manner we can write down a uniform bound for the integrands in the right-hand side of \rf{as6} and \rf{as7}. According to Lemma \ref{l1} and balancing conditions, the product in the right-hand side of \rf{as6} is restricted by
\begin{equation}\label{as10} 
C\exp\eta\biggl(-(n+1)\sum_{j=1}^{n}|v_j|+\sum_{\substack{i,j = 1 \\ i < j}}^n|v_i-v_j|\biggr) \leq 
C'\exp\eta\biggl(-\sum_{j=1}^{n}|v_j|\biggr)\end{equation}
while the product \rf{as7} is restricted by
\begin{equation}\label{as11}
C''\frac{\exp \bigl[-(n+1)\eta(v_{n+1}+(n+1)L)\bigr]}{\exp\bigl[-\eta\sum_{j=1}^n(v_{n+1}-v_j+(n+1)L) \bigr]}=
C''\exp\bigl[-(n+1)L\eta\bigr].
\end{equation}
Here we used the condition
$$\sum_{j=1}^{n+1}v_j=0.$$
The estimates \rf{as10} and \rf{as11} show that the transformed integrals in the left-hand side of \rf{id1}, multiplied by 
\begin{equation} \label{as11a}\exp\frac{2\pi q(n+1)L}{\o_1\o_2} \,
\exp\frac{-\imath \pi}{2\omega_1\omega_2}\biggl(2q(G-F)+\sum_{\ell=1}^{n+m+2}
\left(- g^2_\ell + f^2_\ell\right) \biggr)\end{equation}
have a  convergend majorant and thus tend to the limit equal to the convergent integral
\begin{equation}\label{as12}\frac{1}{n!}\int_{\R^n}
\frac{\prod_{j=1}^{n}\prod_{\ell=1}^{n+m+2}\gamma^{(2)}(g_\ell+\imath v_j\,,f_\ell-\imath v_j) }
{\prod_{1\leq j<k\leq n}\gamma^{(2)}(\pm\imath (v_j- v_k))} \prod_{j=1}^n \frac{dv_j}{ \sqrt{\omega_1\omega_2}}\end{equation}
\cbk

Next perform the same calculation in the $m$-integral in
the right-hand side of \eqref{id1}. The shifts of external variables
\begin{align}
	g_{\ell} \to g_{\ell} +\imath L, \qquad  f_{\ell} \to f_{\ell} -\imath L
\end{align}
induce different shifts of $g'_{\ell}$ and $f'_{\ell}$.
Due to the relations \rf{balans1} we have 
\begin{align}\label{as13}
	g'_{\ell} \to g'_{\ell} + \frac{n+1}{m+1} \, \imath L, \qquad 
	f'_{\ell} \to f'_{\ell} - \frac{n+1}{m+1}\imath L.
\end{align} 
Repeating the same steps we conclude that the integral in the right-hand side of \rf{id1}, multiplied by \rf{as11a} tends to the convergent integral 
\begin{equation}\label{as17}\frac{1}{m!}\int_{\R^m}
\frac{\prod_{j=1}^{m}\prod_{\ell=1}^{n+m+2}\gamma^{(2)}(g'_\ell+\imath u_j\,,f'_\ell-\imath u_j) }
{\prod_{1\leq j<k\leq m}\gamma^{(2)}(\pm\imath (u_j- u_k))} \prod_{j=1}^m \frac{du_j}{ \sqrt{\omega_1\omega_2}}.
\end{equation}

After all we obtain the  relation
\begin{multline}\label{id2}
	\frac{1}{n!}\int_{\R^n}
	\frac{\prod_{j=1}^{n}\prod_{\ell=1}^{n+m+2}\gamma^{(2)}(g_\ell+\imath u_j\,,f_\ell-\imath u_j) }
	{\prod_{1\leq j<k\leq n}\gamma^{(2)}(\pm\imath(u_j- u_k))
	} \prod_{j=1}^n \frac{du_j}{  \sqrt{\omega_1\omega_2}} = \prod_{j,k=1}^{n+m+2}\gamma^{(2)}(g_j+f_k)\\[8pt]
	\times \frac{1}{m!}\int_{\R^m}
	\frac{\prod_{j=1}^{m}\prod_{\ell=1}^{n+m+2}\gamma^{(2)}(g'_\ell+\imath u_j,f'_\ell-\imath u_j)}
	{\prod_{1\leq j<k\leq m}\gamma^{(2)}(\pm\imath(u_j- u_k))
	} \prod_{j=1}^m \frac{du_j}{  \sqrt{\omega_1\omega_2}}
\end{multline} 
valid under balancing conditions \rf{balans}, \rf{balans1}. \cbk

\subsection{First reduction}
\def\iL{\imath L}
In what follows we perform some reductions of the relation \rf{id2}, where the external parameters
obey the balancing condition \rf{balans} and are connected by the relation \rf{balans1}. 

Let us perform the shifts
\begin{align}
	g_{n+m+2} \to g_{n+m+2} -\iL, \qquad f_{n+m+2} \to f_{n+m+2} +\iL
\end{align}
which are compatible with balancing condition. Due to the relations \eqref{balans1} we have
\begin{equation}\label{a24}\begin{split}
	&g'_{n+m+2} \to g'_{n+m+2} - \frac{\iL}{ m+1}+\iL, \qquad 
	f'_{n+m+2} \to f'_{n+m+2} + \frac{\iL}{ m+1}-\iL\\[10pt]
	&g'_{\ell} \to g'_{\ell} - \frac{\iL}{ m+1},\qquad
	f'_{\ell} \to f'_{\ell} +\frac{\iL}{ m+1}, \qquad  \ell=1,\ldots, n+m+1.
\end{split}\end{equation}
Next we calculate calculate the leading asymptotics of both sides of the identity \eqref{id2} as $L \to \infty$.
Denote for simplicity
\begin{align}
	g_{n+m+2} = a, \qquad f_{n+m+2} = b, \qquad g'_{n+m+2} = a', \qquad f'_{n+m+2} = b'.
\end{align}
We have the following pointwise limit in left-hand side as $L \to \infty$
\begin{equation}\label{as31}\begin{split}
	&\prod_{j=1}^{n}\gamma^{(2)}(a + \imath u_{j} -\iL\,,b-\imath u_{j}+\iL)
	\to \\
	&\prod_{j=1}^{n} e^{\frac{\imath \pi}{2\omega_1\omega_2}
		\left((a + \imath u_{j} -\iL -q)^2 -(b-\imath u_{j}+\iL -q)^2\right)}  =
	e^{\frac{\imath \pi}{2\omega_1\omega_2} I_1}
\end{split}\end{equation}
where
\begin{equation}
\begin{aligned}
I_1 = & \sum_{j=1}^{n}
\left(a+b-2q\right)
\left(a-b+2\imath u_j -2\iL\right) \\
&= \left(a+b-2q\right)\left(a-b-2\iL\right)n +
2\left(a+b-2q\right)\sum_{j=1}^{n} \imath u_j
\end{aligned}
\end{equation}
In right-hand side of the identity \rf{id2} we have to shift all the integration
variables 
\begin{equation}\label{as25}u_j \to u_j +\frac{L}{m+1}\end{equation} 
to remove $L$-dependence in almost all functions except containing
$g'_{n+m+2} = a'$ and $f'_{n+m+2} = b'$, so that
\begin{equation}
\begin{aligned}
	&\prod_{j=1}^{m}\gamma^{(2)}(a' + u_{j} +\iL\,,b'-u_{j}-\iL)
	\to\\
	&\prod_{j=1}^{n} e^{-\frac{\imath \pi}{2\omega_1\omega_2}
		\left((a' + \imath u_{j} +\iL -q)^2 -(b'-\imath u_{j}-\iL -q)^2\right)}  =
	e^{\frac{\imath \pi}{2\omega_1\omega_2} I_2}
\end{aligned}
\end{equation}
where
\begin{equation}
\begin{aligned}
I_2 &= \sum_{j=1}^{m}
\left(2q-a'-b'\right)
\left(a'-b'+2\imath u_j +2\iL\right) \\ 
&=\left(2q-a'-b'\right)\left(a'-b'+2\iL\right)m +
2\left(a'+b'-2q\right)\sum_{j=1}^{m} \imath u_j\\
&= (a+b){\frac{m}{m+1}}(G-F) +2\iL(a+b)m -
\left(a^2-b^2\right)m +2(a+b)\sum_{j=1}^{m} \imath u_j.
\end{aligned}
\end{equation}
We also have $L$-dependence in prefactor in right-hand side of \rf{id2} and
\begin{equation}\label{as25a}\begin{split}
	&\prod_{\ell=1}^{n+m+1}\gamma^{(2)}(a -\iL+f_{\ell}\,,b +\iL +g_{\ell})
	\to\\
	&\prod_{\ell=1}^{n+m+1} e^{\frac{\imath \pi}{2\omega_1\omega_2}
		\left((a -\iL+f_{\ell} -q)^2 -(b +\iL +g_{\ell} -q)^2\right)}  =
	e^{\frac{\imath \pi}{2\omega_1\omega_2} I_3}
\end{split}\end{equation}
where
\begin{equation}
\begin{aligned}
I_3 &= \sum_{\ell=1}^{n+m+1}
\left(a+b+f_{\ell}+g_{\ell}-2q\right)
\left(a-b+f_{\ell}-g_{\ell} -2\iL\right) \\[6pt]
&= \left((a+b)(n+m)-2q n\right)\left(a-b-2L\right)
+\left(a+b-2q\right)\left(a-b\right)\\[6pt]
&+ \left(a+b-2q\right)\left(F-G\right) +
\sum_{\ell=1}^{n+m+1}
\left(f_{\ell}+g_{\ell}\right)
\left(f_{\ell}-g_{\ell}\right).
\end{aligned}
\end{equation}
We have
\begin{equation}
\begin{aligned}
I_2+I_3 &=
-2\iL\left(a+b-2q\right)n + \left(a+b-2q\right)\left(a-b\right)(n+1)\\[8pt]
&+\left(a+b-2q\right)\left(F-G\right) +(a+b){\frac{m}{m+1}}(G-F) \\[6pt]
&+	\sum_{\ell=1}^{n+m+1}
\left(f_{\ell}+g_{\ell}\right)
\left(f_{\ell}-g_{\ell}\right)+2(a+b)\sum_{j=1}^{m} \imath u_j.
\end{aligned}
\end{equation}
The $L$-dependence in $I_1$ and $I_2+I_3$ is the same so that
asymptotic behaviour of both sides of the identity is the same and in the limit we arrive at
\begin{equation}
\begin{split}\label{fgF}
	\frac{1}{n!}\int_{\R^n}\,
	e^{\frac{\pi}{\omega_1\omega_2}\left(2q-a-b\right)\sum_{j=1}^{n} u_j}\,
	\frac{\prod_{j=1}^{n}\prod_{\ell=1}^{n+m+1}\gamma^{(2)}(g_\ell+ \imath u_j\,,f_\ell-\imath u_j) }
	{\prod_{1\leq j<k\leq n}\gamma^{(2)}(\pm\imath(u_j- u_k))
	} \prod_{j=1}^n \frac{ du_j}{\sqrt{\o_1\o_2}} \\[6pt]
	= e^{\frac{\imath \pi}{2\omega_1\omega_2}\,\varphi(a,b,f,g)}\,\gamma^{(2)}(a+b)\,
	\prod_{j,k=1}^{n+m+1}\gamma^{(2)}(g_j+f_k) \\[6pt]
	\times \frac{1}{m!}\int_{\R^m}\,
	e^{\frac{\pi}{\omega_1\omega_2}\left(a+b\right)\sum_{j=1}^{m} u_j}\,
	\frac{\prod_{j=1}^{m}\prod_{\ell=1}^{n+m+1}\gamma^{(2)}(g'_\ell+\imath u_j,f'_\ell-\imath u_j)}
	{\prod_{1\leq j<k\leq m}\gamma^{(2)}(\pm\imath (u_j- u_k))
	} \prod_{j=1}^m\frac{ du_j}{\sqrt{\o_1\o_2}}
\end{split}\end{equation}
provided we are able to obtain the uniform bounds for the corresponding integrands. Here we introduced the function
\begin{equation}
\begin{aligned}
\varphi(a,b,f,g) &= \left(a+b-2q\right)\left(F-G+a-b\right) +
(a+b){\frac{m}{m+1}}(G-F) +\sum_{\ell=1}^{n+m+1}
\left(f^2_{\ell}-g^2_{\ell}\right) \\[6pt]
&= (a^2-b^2){\frac{m}{m+1}} + \left({\frac{a+b}{m+1}}-2q\right)
\sum_{\ell=1}^{n+m+1} \left(f_{\ell}-g_{\ell}\right) +
\sum_{\ell=1}^{n+m+1} \left(f^2_{\ell}-g^2_{\ell}\right).
\end{aligned}
\end{equation}
Estimate first the nominator of the integrand in the left-hand side of \rf{id2}. Collect its factors
containing the variable $u_j$. They are equal to 
\begin{equation}\label{as18}
G_j= \g^{(2)}(a+\imath u_j-\imath L, b-\imath u_j+\imath L)\prod_{l=1}^{n+m+1}\g^{(2)}(g_l+\imath u_j,f_l-\imath u_j).\end{equation}
Due to Lemma \ref{l1} 
\begin{equation} \label{as19}\begin{split}
	|G_j|&<C_j\exp \,\Re \frac{\pi}{\o_1\o_2}\biggl(-\sum_{l=1}^{n+m+1}(2q-g_l-f_l)|u_j|-(2q-a-b)|u_j-L|\biggr)\\[8pt]
	& = \exp \,\Re \frac{\pi}{\o_1\o_2}\bigl( -(2qn+a+b)|u_j|-(2q-a-b)|u_j-L|\big).
\end{split}\end{equation}
Assuming the condition \rf{as0a} for parameters $a$ and $b$ we see that the last line of \rf{as19}
is represented by exponent of the piecewise linear function $-\delta(u_j)$, where
$$\delta(u_j)=\a|u_j|+\b|u_j-L|$$
with positive coefficients $\a$ and $\b$, $\a>\b$. Elementary analysis of its graph shows that
\begin{equation}\label{as20} \delta(u_j)>\b L+(\a-\b)|u_j|. \end{equation} 
Thus, we have the bound 
\begin{equation}\label{as21} |G_j|<C_j \exp \left(-\pi \Re\frac{2q-a-b}{\o_1\o_2}L -2\pi\Re\frac{q(n-1)+a+b}{\o_1\o_2}|u_j|\right). \end{equation}
Multiplying over all $j$ we arrive to the desired asymptotics
\begin{equation}
\exp \biggl( -\pi n\, \Re \frac{2q - a - b}{\o_1\o_2} L \biggr)\end{equation}
multiplied by the integral with integrand uniformly bounded by 
\begin{multline}
	C \exp \, \Re\frac{2\pi}{\o_1\o_1}\biggl(- \sum_{j=1}^{n}\bigl(q(n-1)+(a+b)\bigr)|u_j|+q
	\sum_{\substack{i,j=1 \\ i<j}}^n |u_i-u_j|\biggr)\\[6pt]
	 \leq C'\exp \biggl( -\Re \frac{2\pi(a+b)}{\o_1\o_2}\sum_{j = 1}^n|u_j| \biggr).
\end{multline}
The latter absolutely converges once
\begin{equation}\label{as24} \Re \frac{a+b}{\o_1\o_2}>0.\end{equation}
In the same manner we estimate the integral in the right-hand side of \rf{id2}. Collect all factors of the nominator  containing the shifted variable $u_j$ into the product $G'_j$
\begin{equation}\label{as23}
G'_j= \g^{(2)}(a'+\imath u_j-\imath L, b'-\imath u_j+\imath L)\prod_{l=1}^{n+m+1}\g^{(2)}(g'_l+\imath u_j,f'_l-\imath u_j).\end{equation}
Following the same lines as before we see that the integrand is a product of its asymptotics
\begin{equation}\label{as22}\exp \biggl(-\frac{\pi m}{\o_1\o_2}(2q-a'-b')L\biggr)=\exp \biggl(-\frac{\pi m}{\o_1\o_2}(a+b)L\biggr)\end{equation}
multiplied by the function which can be estimated by a uniform absolutely integrable function
\begin{multline} 
	\exp \, \Re \frac{2\pi}{\o_1\o_1}\biggl( - \sum_{j=1}^{m}\bigl(q(n-1)+(a'+b')\bigr)|u_j|+q
	\sum_{\substack{i,j=1 \\ i<j}}^m|u_i-u_j|\biggr)\\[6pt]
	\leq C'\exp \biggl( -\Re \frac{2\pi(a'+b')}{\o_1\o_2}\sum_{j = 1}^m|u_j| \biggr)=C'\exp \biggl(-\Re \frac{2\pi(2q-a-b)}{\o_1\o_2}\sum_{j = 1}^m|u_j|  \biggr).
\end{multline}
Combining this bound with the limit \rf{as25a} we also see that the right-hand side of \rf{id2} divided by its asymptotics is given by uniformly bounded integral. This finishes the proof of the relation \rf{fgF}.

\medskip

We can write down the relation \rf{fgF} in a slightly different form by separating $f$ and $g$-dependence in the function $\varphi(a,b,f,g)$. Namely, denote
\begin{equation}
\vf(g) =
\left({\frac{a+b}{m+1}}-2q\right)
\sum_{\ell=1}^{n+m+1} g_{\ell} +
\sum_{\ell=1}^{n+m+1} g^2_{\ell}.
\end{equation}
Then
\begin{equation}\label{as29}
	\varphi(a,b,f,g) = (a^2-b^2){\frac{m}{m+1}}-\vf(g)+\vf(f).
\end{equation}
The external parameters $g_{\ell}$ and $f_{\ell}$ obey the following
balancing condition
\begin{align}
	g+f+(a+b) = 2q(m+1), \qquad 2q = \omega_1+\omega_2
\end{align}
where 
\begin{equation}
g=\sum_{\ell=1}^{n+m+1}g_\ell, \qquad f=\sum_{\ell=1}^{n+m+1}f_\ell.
\end{equation}
Parameters $g'_{\ell}$ and $f'_{\ell}$ are connected with
$g_{\ell}$ and $f_{\ell}$ by simple transformation
\begin{equation}\label{gftick}
g'_{\ell} =\frac{g+a}{ m+1}-g_\ell, \qquad 
f'_{\ell} =\frac{f+b}{ m+1}-f_\ell, \qquad \ell=1,\ldots, n+m+1
\end{equation}
and in the same way
\begin{equation}
a' =\frac{g+a}{ m+1}-a, \qquad b' =\frac{f+b }{m+1}-b, \qquad a'+b' = 2q-a-b.
\end{equation}
Finally, using the notation \rf{as29} we rewrite the  relation \rf{fgF} as
\begin{equation}\label{fg}\begin{split}
	e^{\vf(g)}\,\frac{1}{n!}\int_{\R^n}\,
	e^{\frac{\pi}{\omega_1\omega_2}\left(2q-a-b\right)\sum_{j=1}^{n} u_j}\,
	\frac{\prod_{j=1}^{n}\prod_{\ell=1}^{n+m+1}\gamma^{(2)}(g_\ell+ \imath u_j\,,f_\ell-\imath u_j) }
	{\prod_{1\leq j<k\leq n}\gamma^{(2)}(\pm\imath(u_j- u_k))
	} \prod_{j=1}^n \frac{du_j}{\sqrt{\o_1\o_2}} \\[8pt]
	= \gamma^{(2)}(a+b)\,e^{\frac{\imath \pi}{2\omega_1\omega_2}\,(a^2-b^2){\frac{m}{m+1}}}\,
	\prod_{j,k=1}^{n+m+1}\gamma^{(2)}(g_j+f_k)\\[8pt]
	\times e^{\vf(f)}\,\frac{1}{m!}\int_{\R^m}\,
	e^{\frac{\pi}{\omega_1\omega_2}\left(a+b\right)\sum_{j=1}^{m} u_j}\,
	\frac{\prod_{j=1}^{m}\prod_{\ell=1}^{n+m+1}\gamma^{(2)}(g'_\ell+ \imath u_j,f'_\ell-\imath u_j)}
	{\prod_{1\leq j<k\leq m}\gamma^{(2)}(\pm\imath(u_j- u_k))
	} \prod_{j=1}^m \frac{du_j}{\sqrt{\o_1\o_2}}.
\end{split}\end{equation}
This relation is valid under the conditions \rf{as0a} and \rf{as24}.

We should note that a special case of obtained formula is presented in the forthcoming paper \cite{SS}.

\subsection{Second reduction}
Now we set $m=n$ and use the following parametrization
\begin{equation}
	\begin{aligned}
	&g_{k} = -\imath x_k+\frac{g^*}{2}, \qquad  g_{n+k} = - \imath z_k+\frac{g}{2},\\
	& f_{k} = \imath x_k+\frac{g^*}{2}, \qquad f_{n+k} = \imath z_k+\frac{g}{2}
	\end{aligned}
\end{equation}
for $k = 1, \ldots, n$ and
\begin{equation}
g_{2n+1} = q-a + \imath \sum_{k=1}^n (x_k+z_k), \qquad f_{2n+1} = q-b - \imath \sum_{k=1}^n (x_k+z_k).
\end{equation}
Then since $g + g^* = 2q$ from the relations \eqref{gftick} we also have
\begin{equation}
\begin{aligned}
	& g'_{k} = \imath x_k+\frac{g}{2}, \qquad g'_{n+k} = iz_k+\frac{g^*}{2}, \\
	& f'_{k} = -ix_k+\frac{g}{2}, \qquad f'_{n+k} = -iz_k+\frac{g^*}{2}
\end{aligned}
\end{equation}
for $k = 1, \ldots, n$ and
\begin{equation}
g'_{2n+1} = a - \imath \sum_{k=1}^n (x_k+z_k), \qquad f'_{2n+1} = b + \imath \sum_{k=1}^n (x_k+z_k).
\end{equation}
The product behind the integral in the right-hand side of \eqref{fg}
\begin{align*}
	\prod_{j,k=1}^{2n+1}\gamma^{(2)}(g_j+f_k) & = 
	\gamma^{(2)}(g_{2n+1}+f_{2n+1})\,
	\prod_{k=1}^{2n}\gamma^{(2)}(g_{2n+1}+f_k)\, \\
& \times	\prod_{k=1}^{2n}\gamma^{(2)}(f_{2n+1}+g_k)\,
	\prod_{j,k=1}^{2n}\gamma^{(2)}(g_j+f_k).
\end{align*}
Then the relation \rf{fg} takes the form
\begin{equation} \label{as36}\begin{split}
	& \int_{\R^n}\!\!
	e^{\frac{\pi(2q-a-b)}{\omega_1\omega_2}\sum\limits_{j=1}^{n} u_j}\,
	\!\prod_{j=1}^{n}\textstyle\gamma^{(2)}\Bigl(q-a + \imath \sum\limits_{k=1}^n (x_k+z_k)+ \imath u_j, q-b - \imath \sum\limits_{k=1}^n (x_k+z_k)-\imath u_j \Bigr) \\[10pt]
	& \times \frac{\prod_{j=1}^{n}\prod_{k=1}^{n}
		\gamma^{(2)}\left(\pm\imath(x_k-u_j)+\frac{g^*}{2}\,,\pm\imath(z_k-u_j)+\frac{g}{2}\right) }
	{\prod_{1\leq j<k\leq n}\gamma^{(2)}(\pm\imath(u_j- u_k))
	} \prod_{j=1}^n \frac{du_j}{\sqrt{\o_1\o_2}} = H(x,z,a,b)  \\[10pt]
	& \times \int_{\R^n}
	e^{\frac{\pi(a+b)}{\omega_1\omega_2}\sum\limits_{j=1}^{n} u_j}\,
	\prod_{j=1}^{n}\textstyle\gamma^{(2)}\Bigl(a - \imath \sum\limits_{k=1}^n (x_k+z_k)- \imath u_j\,, b + \imath \sum\limits_{k=1}^n (x_k+z_k)+\imath u_j\Bigr) \\[10pt]
	& \times \frac{\prod_{j=1}^{n}\prod_{k=1}^{n}
		\gamma^{(2)}\left(\pm\imath(x_k-u_j)+\frac{g}{2}\,,\pm\imath(z_k-u_j)+\frac{g^*}{2}\right) }
	{\prod_{1\leq j<k\leq n}\gamma^{(2)}(\pm\imath(u_j- u_k))
	} \prod_{j=1}^n\frac{du_j}{\sqrt{\o_1\o_2}}
\end{split}\end{equation}
where
\begin{equation}\label{as33}\begin{split}
	&H(x,z,a,b)=	e^{\frac{\pi}{\omega_1\omega_2}\,
		\left((2q-a-b)\sum\limits_{k=1}^n (x_k+z_k) - g^{*}\sum\limits_{k=1}^n x_k-
		g\sum\limits_{k=1}^n z_k\right)} \\[6pt]
	&\times	\prod_{\substack{j,k = 1 \\ j\neq k}}^{n}\gamma^{(2)}(\imath(x_k-x_j)+g^*)\,
	\gamma^{(2)}(\imath(z_k-z_j)+g)\\[6pt]
	&\times	\prod_{k=1}^{n}\textstyle\gamma^{(2)}\Bigl(q-a + \imath \sum\limits_{k=1}^n (x_k+z_k)+\imath x_k+\frac{g^*}{2}\Bigr)\,
	\gamma^{(2)}\Bigl(q-a + \imath \sum\limits_{k=1}^n (x_k+z_k)+\imath z_k+\frac{g}{2}\Bigr)\\[6pt]
	&\times \prod_{k=1}^{n}\textstyle\gamma^{(2)}\Bigl(q-b - \imath \sum\limits_{k=1}^n (x_k+z_k)-\imath x_k+\frac{g^*}{2}\Bigr)\,
	\gamma^{(2)}\Bigl(q-b - \imath \sum\limits_{k=1}^n (x_k+z_k)-\imath z_k+\frac{g}{2}\Bigr).
\end{split}\end{equation}
Now we shift $a \to a-\imath L$ and $b \to b+\imath L$ and calculate asymptotics
as $L \to \infty$ using the relation \rf{as31}.

In the left-hand side of \rf{as36} we have
\begin{equation}
\begin{aligned}
\prod_{j=1}^{n}\textstyle\gamma^{(2)} \Bigl(q-a + \imath \sum\limits_{k=1}^n (x_k+z_k)+ \imath u_j+\imath L,
q-b - \imath \sum_{k=1}^n (x_k+z_k)-\imath u_j-\imath L \Bigr) \to \\ e^{\frac{\imath \pi}{2\omega_1\omega_2} I_1},\qquad
I_1 = 2\imath(a+b)\Big(nL+\sum_{j=1}^n u_j\Big) +n(a+b)\Big(2 \imath \sum_{k=1}^n (x_k+z_k) -(a-b)\Big).
\end{aligned}
\end{equation}
In right-hand side of \rf{as36} in the integrand we have
\begin{equation}
\begin{aligned}
&\prod_{j=1}^{n}\textstyle\gamma^{(2)} \Bigl(a - \imath \sum\limits_{k=1}^n (x_k+z_k)- \imath u_j-\imath L,
b + \imath \sum\limits_{k=1}^n (x_k+z_k)+\imath u_j+\imath L \Bigr) \to e^{\frac{\imath \pi}{2\omega_1\omega_2} I_2},\\
&I_2 = 2\imath(2q-a-b)\Big(nL+\sum_{j=1}^n u_j\Big)- n(a+b-2q)
\Big(2 \imath \sum_{k=1}^n (x_k+z_k) -(a-b)\Big)
\end{aligned}
\end{equation}
and for the two factors outside of the integral we have
\begin{equation}
\begin{aligned}
&\prod_{k=1}^{n}\textstyle
\gamma^{(2)}\Bigl(q - a + \imath \sum\limits_{k=1}^n (x_k+z_k)+ \imath x_k+\frac{g^*}{2}+\imath L, \\[6pt]
&\hspace{1.6cm} \textstyle q-b -\imath \sum\limits_{k=1}^n (x_k + z_k)- \imath x_k+\frac{g^*}{2} -\imath L\Bigr)
\to  e^{\frac{\imath \pi}{2\omega_1\omega_2} I_3},\\[6pt]
&I_3 = (a+b-g^\ast)\Big(2\imath n L+2\imath\sum_{k=1}^n\big((n+1)x_k+nz_k\big)-n(a-b)\Big) , 
\end{aligned}
\end{equation}
and
\begin{align*}
	&\prod_{k=1}^{n}\textstyle
	\gamma^{(2)}\Big(q - a + \imath \sum\limits_{k=1}^n (x_k + z_k) + \imath z_k+ \frac{g}{2}+\imath L, \\[6pt]
	&\hspace{1.6cm} \textstyle q-b - \imath \sum\limits_{k=1}^n (x_k + z_k) - \imath z_k+\frac{g}{2}-\imath L\Bigr)
	\to e^{\frac{\imath \pi}{2\omega_1\omega_2} I_4},\\[6pt]
	&I_4 =(a+b-g)\Big(2\imath n L+2\imath\sum_{k=1}^n\big(nx_k+(n+1)z_k\big)-n(a-b)\Big).
\end{align*}
Collecting all these calculations we see that integrands from both sides have equal asymptotics
\begin{equation}\label{as34} \exp\Big(-\frac{\pi nL}{\o_1\o_2}(a+b)\Big)\end{equation}
while the rest has a poinwise limit,
so that the initial relation is reduced to the following equality
\begin{equation}\label{IM}\begin{split}
	\int_{\R^n}
	e^{\frac{2\pi \lambda}{\omega_1\omega_2}\sum\limits_{j=1}^{n} u_j}\,
	\frac{\prod_{j=1}^{n}\prod_{k=1}^{n}
		\gamma^{(2)}\left(\pm\imath(x_k-u_j)+\frac{g^*}{2}\,,\pm\imath(z_k-u_j)+\frac{g}{2}\right) }
	{\prod_{1\leq j<k\leq n}\gamma^{(2)}(\pm\imath(u_j- u_k))
	} \prod_{j=1}^n du_j \\[6pt]
	= e^{\frac{\pi\lambda}{\omega_1\omega_2}\sum\limits_{k=1}^n (x_k+z_k)}\,
	\prod_{\substack{j,k = 1 \\ j\neq k}}^{n}\gamma^{(2)}(\imath(x_k-x_j)+g^*)\,
	\gamma^{(2)}(\imath(z_k-z_j)+g) \\[6pt]
	\times \int_{\R^n}
	e^{-\frac{2\pi\lambda}{\omega_1\omega_2}\sum\limits_{j=1}^{n} u_j}\,
	\frac{\prod_{j=1}^{n}\prod_{k=1}^{n}
		\gamma^{(2)}\left(\pm\imath(x_k-u_j)+\frac{g}{2}\,,\pm\imath(z_k-u_j)+\frac{g^*}{2}\right) }
	{\prod_{1\leq j<k\leq n}\gamma^{(2)}(\pm\imath(u_j- u_k))
	} \prod_{j=1}^n du_j
\end{split}\end{equation}
where
\begin{equation} \label{as35} \l=q-a-b,\qquad \Big|\Re \frac{\l}{\o_1\o_2}\Big|<\frac{\eta}{2}=\frac{1}{2} \Re \big(\o_1^{-1}+\o_2^{-1}\big)\end{equation}
provided we can obtain the uniform integrable bounds for the integrands divided by asymptotics  \rf{as34}.

Let us estimate the integrand of the left-hand side of \rf{as36}. Due to Lemma \ref{l1} its nominator is bounded by
\begin{equation}\label{as37}
C\prod_{j=1}^n \exp \Re \gamma A_j(u_j),\qquad \gamma= \frac{\pi}{\o_1\o_2}
\end{equation}
with some constant $C$ and 
\begin{equation} \label{as38} A_j(u_j)= (2q-a-b)u_j-g\sum_{k=1}^n|u_j-x_k|-g^\ast\sum_{k = 1}^n|u_j-z_k|-(a+b) \Bigl|u_j+L-\sum_{k = 1}^n(x_k+z_k) \Bigr|
\end{equation}
Using triangle inequalities 
$$- |u_j-x_k|\leq -|u_j|+|x_k|,\qquad - |u_j-z_k| \leq -|u_j|+|z_k|$$
we can replace $A_j(u_j)$ by $\tilde{A}_j(u_j)+C'_j$, where $C'_j$ is some constant and
\begin{equation}
\tilde{A_j}(u_j)= (2q-a-b)u_j-2qn|u_j|-(a+b)|u_j+L| \end{equation}
The function $\xi(x)=-\Re \gamma\tilde{A_j}(x)$ is a piecewise linear function of the form
\begin{equation} \label{as39}\xi(x)=  n(\a+\b)|x|+\b|x+L|-\a x \end{equation}
where $\a,\b>0$ are given by \rf{as0a}. 
Analysing the graph of this function, we get inequality
\begin{equation}\label{as41}
\xi(x)\geq\beta L+\big((n-1)(\a+\beta)+2\min(\a,\beta)|x|\big).\end{equation}
This inequality implies the following bound for the integrand in the left-hand side of \rf{as36} divided by its asymptotics
\rf{as34}
\begin{equation} \label{as42}\begin{split} 
	C\exp \eta &\Big(-(n-1)-\min(\a,\b)\sum_{j=1}^n|u_j|+\sum_{\substack{i,j = 1 \\ i<j}}^n|u_i-u_j|\Big)\\
	& \leq C' \exp \biggl(-\eta\min(\a,\b)\sum_{j=1}^n|u_j|\biggr).\end{split}\end{equation}
The latter is absolutely integrable function. The right-hand side is analysed in a similar manner.

\setcounter{equation}{0}
\section{Some inequalities}\label{AppC}
In our previous paper we proved the following little lemma \cite[Lemma 1]{BDKK2}.
\begin{lemma*}
	For any $\epsilon \in [0, 2]$, $y_1, y_2, y \in \mathbb{R}$ we have
	\begin{equation}\label{ineq}
	|y_1 - y_2| - |y_1 - y| - |y_2 - y| \leq \epsilon \left( |y_1| + |y_2| - |y| \right).
	\end{equation}
\end{lemma*}
Now with its help we prove one inequality used in the main text. Define 
\begin{equation}
L_n(\by_{n - 1}, \bx_n) = \sum_{ \substack{ i, j = 1 \\ i < j} }^n | x_i - x_j | + \sum_{ \substack{ i, j = 1 \\ i < j} }^{n - 1} | y_ i - y_j| - \sum_{i = 1}^{n}\sum_{j = 1}^{n - 1} |x_i - y_j|.
\end{equation}
As before, by $\| \bx_n \|$ denote $L^1$-norm. The following statement implicitly appeared during the proof of Lemma 2 in \cite{BDKK2}.
\begin{lemma}
	For any $\ve \in [0, 2]$ we have
	\begin{equation}\label{Ln-bound}
	L_n \leq (n - 1)  \ve \| \bx_n \| - \ve \| \by_{n - 1} \|.
	\end{equation}
\end{lemma}
\begin{proof}
	Both sides of the stated inequality are symmetric with respect to components of $\bx_n, \by_{n - 1}$. Therefore, without loss of generality we assume the ordering
	\begin{equation}\label{ord}
	x_1 \geq \ldots \geq x_n, \qquad y_1 \geq \ldots \geq y_{n - 1}.
	\end{equation}
	For the vector $\bx_n$ with ordered components we write
	\begin{equation}\label{ord-comp}
	\sum_{\substack{i, j = 1 \\ i < j}}^n | x_i - x_j | =  \sum_{m = 1}^{\lfloor n/2 \rfloor} (n - 2m + 1) | x_m - x_{n - m + 1} |.
	\end{equation}
	Similarly for $\by_{n - 1}$.
	Consequently, 
	\begin{equation}\label{Sineq}
	\begin{aligned}
	L_n = \sum_{m = 1}^{\lfloor n/2 \rfloor} (n - 2m + 1) | x_m - x_{n - m + 1} | &+ \sum_{m = 1}^{\lfloor (n - 1)/2 \rfloor} (n - 2m) | y_m - y_{n - m} |  \\[5pt]
	&- \sum_{i = 1}^{n}\sum_{j = 1}^{n - 1} |x_i - y_j|.
	\end{aligned}
	\end{equation}
	Next step it to regroup terms. Consider term with $m = 1$ from the first sum and terms with $i = 1, n$ from the third double sum and write the estimate
	\begin{equation}\label{ineq1}
	\begin{aligned}
	(n - 1) | x_1 - x_n | &- \sum_{j = 1}^{n - 1}  \left( | x_1 - y_j | + | x_n - y_j | \right) \\[4pt]
	& \leq (n - 1) \, \ve \left( | x_1 | + |x_n | \right) - \ve \, \| \bm{y}_{n - 1} \|,
	\end{aligned}
	\end{equation}
	where we used inequality \eqref{ineq} multiple times. Similarly let us estimate the term with $m > 1$ from the first sum together with the corresponding terms from the third double sum
	\begin{equation}\label{ineq2}
	(n - 2m + 1) | x_m - x_{n - m + 1} | - \sum_{j = m}^{n - m} \left( | x_m - y_j | + | x_{n - m + 1} - y_j | \right) \leq 0,
	\end{equation}
	where we used triangle inequality multiple times. Remaining from the third double sum terms can be grouped with terms from the second sum
	\begin{equation}\label{ineq3}
	(n - 2m)  | y_m - y_{n - m} | - \sum_{i = m + 1}^{n - m} \left( | x_i - y_m | + | x_i - y_{n - m} | \right) \leq 0,
	\end{equation}
	where we again used triangle inequalities. Collecting everything together we have
	\begin{equation}
	L_n \leq (n - 1) \, \ve \left( | x_1 | + |x_n | \right) - \ve \, \| \bm{y}_{n - 1} \| \leq (n - 1) \, \ve \| \bx_n \| - \ve \, \| \bm{y}_{n - 1} \|.
	\end{equation}
\end{proof}


\begin{thebibliography}{99}

\bibitem[B]{B} E. W. Barnes, \textit{The theory of the double gamma function}, Philosophical Transactions of the Royal Society of London. Series A, Containing Papers of a Mathematical or Physical Character \textbf{196}, 265--387  (1901).	
%
\bibitem[BDKK]{BDKK} N. Belousov, S. Derkachov, S. Kharchev, S. Khoroshkin, \textit{Baxter operators in Ruijsenaars hyperbolic system I. Commutativity of $Q$-operators}, arXiv:2303.06383 (2023).

\bibitem[BDKK2]{BDKK2} N. Belousov, S. Derkachov, S. Kharchev, S. Khoroshkin, \textit{Baxter operators in Ruijsenaars hyperbolic system II. Bispectral wave functions}, arXiv:2303.06382 (2023).

\bibitem[BDKK3]{BDKK3} N. Belousov, S. Derkachov, S. Kharchev, S. Khoroshkin, \textit{Baxter operators in Ruijsenaars hyperbolic system III. Orthogonality and completeness of wave functions}, arXiv:2307.16817 (2023).
%
\bibitem[GLO]{GLO} A. Gerasimov, D. Lebedev, S. Oblezin, \textit{Baxter operator formalism for Macdonald polynomials}, Letters in Mathematical Physics \textbf{104}, 115--139 (2014).
%
\bibitem[F]{F} L. D. Faddeev, \textit{Discrete Heisenberg-Weyl Group and modular group}, Letters in Mathematical Physics \textbf{34}, 249–254  (1995).
%
\bibitem[FKV]{FKV} L. D. Faddeev, R. M. Kashaev, A. Yu. Volkov, \textit{Strongly Coupled Quantum Discrete Liouville Theory. I: Algebraic Approach and Duality}, Communications in Mathematical Physics \textbf{219}:1, 199--219  (2001).
%
\bibitem[HR1]{HR0} M. Halln\"as, S. Ruijsenaars, \textit{Kernel functions and Bäcklund transformations for relativistic Calogero-Moser and Toda systems}, Journal of Mathematical Physics \textbf{53}:12 (2012), 123512.
%
\bibitem[HR2]{HR1} M. Halln\"as, S. Ruijsenaars, \textit{Joint Eigenfunctions for the Relativistic Calogero–Moser Hamiltonians of Hyperbolic Type: I. First Steps}, International Mathematics Research Notices \textbf{2014}:16, 4400--4456  (2014).
%
\bibitem[HR3]{HR2} M. Halln\"as, S. Ruijsenaars, \textit{Joint Eigenfunctions for the Relativistic Calogero–Moser Hamiltonians of Hyperbolic Type II. The Two-and Three-Variable Cases}, International Mathematics Research Notices \textbf{2018}:14, 4404--4449  (2018).
%
\bibitem[HR4]{HR3} M. Halln\"as, S. Ruijsenaars, \textit{Joint eigenfunctions for the relativistic Calogero–Moser Hamiltonians of hyperbolic type. III. Factorized asymptotics}, International Mathematics Research Notices \textbf{2021}:6, 4679--4708  (2021).
%
\bibitem[KMS]{KMS} V. B. Kuznetsov, V. V. Mangazeev, E. K. Sklyanin, \textit{$Q$-operator and factorised separation chain for Jack polynomials}, Indagationes Mathematicae \textbf{14}:3-4, 451--482 (2003).
%
\bibitem[KS]{KS} V. B. Kuznetsov, E. K. Sklyanin, \textit{On Backlund transformations for many-body systems}, Journal of Physics A \textbf{31}, 2241--2251 (1998).
%
\bibitem[Ku]{Ku} N. Kurokawa, S-Y. Koyama, \textit{Multiple sine functions}, Forum Mathematicum \textbf{15}, 839--876  (2003).
%
\bibitem[N]{N} N. Nekrasov, \textit{Five-dimensional gauge theories and relativistic integrable systems}, Nuclear Physics B \textbf{531}, 1--3 (1998).
%
\bibitem[NS]{NS} M. Noumi, A. Sano, \textit{An infinite family of higher-order difference operators that commute with Ruijsenaars operators of type A}, Letters in Mathematical Physics \textbf{111}:4, 91 (2021).

\bibitem[PT]{PT} B. Ponsot, J. Teschner, \textit{Clebsch–Gordan and Racah–Wigner Coefficients for a Continuous Series of Representations of $ U_q (sl(2, \R))$}, Communications in Mathematical Physics \textbf{224}:3, 613--655 (2001).	
%
\bibitem[Ra1]{Rains1} E. M. Rains, \textit{Transformations of elliptic hypergeometric integrals}, Annals of Mathematics, 169--243  (2010).	

\bibitem[Ra2]{Rains2} E. M. Rains, \textit{Limits of elliptic hypergeometric integrals}, The Ramanujan Journal \textbf{18}:3, 257--306 (2009).

\bibitem[R1]{R1} S. N. M. Ruijsenaars, \textit{Complete integrability of relativistic Calogero-Moser systems and elliptic function identities}, Communications in Mathematical Physics \textbf{110}, 191--213  (1987).
%
\bibitem[R2]{R2} S. N. M. Ruijsenaars, \textit{First-order analytic difference equations and integrable quantum systems}, Journal of Mathematical Physics \textbf{38}, 1069--1146  (1997).
%
\bibitem[R3]{R3} S. N. M. Ruijsenaars, \textit{A relativistic conical function and its Whittaker limits}, SIGMA. Symmetry, Integrability and Geometry: Methods and Applications \textbf{7}, 101  (2011).
%
\bibitem[R4]{R4} S. N. M. Ruijsenaars, \textit{Sine-Gordon solitons vs. relativistic Calogero-Moser particles}, In Proceedings of Integrable structures of exactly solvable two-dimensional models of quantum field theory, 273--292 (2001). 
%
\bibitem[S]{S} V. P. Spiridonov, \textit{Theta hypergeometric integrals}, St. Petersburg Math. J. \textbf{15}, 929--967 (2004).
%
\bibitem[Skl]{Skl} E. K. Sklyanin, \textit{B\"acklund transformations and Baxter’s Q-operator},
Lecture notes, Integrable systems: from classical to quantum, Universite de Montreal (Jul 26 – Aug 6, 1999), nlin/0009009 [nlin.SI].
%
\bibitem[SS]{SS} G. A. Sarkissian, V. P. Spiridonov, \textit{Complex and rational hypergeometric functions on root systems}, to appear.



\end{thebibliography}
\end{document}